\documentclass[journal]{IEEEtran}
\IEEEoverridecommandlockouts     
\usepackage{graphics} 
\usepackage{epsfig} 
\usepackage{mathptmx} 
\usepackage{times} 

\usepackage{amssymb, amsmath, mathrsfs}

\usepackage{amsthm}
\usepackage{epstopdf} 
\usepackage{cite}
\usepackage{thmtools}
\usepackage{etoolbox}
\usepackage{subeqnarray}
\usepackage{cases}
\usepackage{phonetic}
\DeclareMathOperator{\rank}{rank} %
\DeclareMathOperator{\diag}{diag} %
\DeclareMathOperator{\tr}{tr} %
\DeclareMathOperator{\im}{Im} %
\DeclareMathOperator{\re}{Re} %

\DeclareMathAlphabet{\mathpzc}{OT1}{pzc}{m}{it}

\usepackage{calc,pifont} 
\newcounter{local}
\renewcommand\theenumi{\protect\setcounter{local}%
  {171+\the\value{enumi}}\protect\ding{\value{local}}}

\theoremstyle{definition}
\newtheorem{thm}{Theorem}
\newtheorem{lem}{Lemma}
\newtheorem{rmk}{Remark}
\newtheorem{exam}{Example}

\allowdisplaybreaks

\title{\LARGE \bf Pure Gaussian quantum states from passive Hamiltonians and an active local dissipative process}

\author{Shan Ma,~Matthew J. Woolley,~Ian R. Petersen, and Naoki Yamamoto
\thanks{This work was supported by the Australian Research Council (ARC), the Australian Academy of Science, and the Japan Society for the Promotion of Science (JSPS).}
\thanks{S. Ma, M. J. Woolley and I. R. Petersen are with the
School of Engineering and Information Technology, University of New South Wales 
at the Australian Defence Force Academy, Canberra, ACT 2600, Australia. 
        {\tt\small shanma.adfa@gmail.com  m.woolley@adfa.edu.au  i.r.petersen@gmail.com  }}%
\thanks{N. Yamamoto is with the Department of Applied Physics and Physico-Informatics, 
Keio University, Yokohama 223-8522, Japan. 
        {\tt\small yamamoto@appi.keio.ac.jp}}%
}

\begin{document}

\maketitle
\thispagestyle{empty}
\pagestyle{empty}

\begin{abstract}
We investigate the problem of preparing a pure Gaussian state via reservoir engineering. In particular, we consider a linear quantum system with a passive Hamiltonian and with a single reservoir which acts only on a single site of the system. We then give a full parametrization of the pure Gaussian states that can be prepared by this type of quantum system.
\end{abstract}


\section{Introduction}
Recently, the problem of deterministically preparing a pure Gaussian state 
in a linear quantum system has been studied in the literature 
\cite{KY12:pra,Y12:ptrsa,IY13:pra, MWPY14:msc,MWPY16:arxiv,MPW15:arxiv}. 
The main idea is to construct coherent and dissipative processes such that the quantum system 
is strictly stable and  driven into a desired target pure Gaussian state. 
This approach is often referred to as {\it reservoir engineering} 
\cite{PCZ96:prl,WC14:pra}. 
Let us consider a linear open quantum system that obeys a Markovian Lindblad 
master equation~\cite{WM10:book}: 
\begin{align}\label{MME}
        \frac{d}{d t}\hat{\rho} 
            =-i[\hat{H},\; \hat{\rho}]
                +\sum\limits_{j=1}^{K}\left(
                       \hat{c}_{j}\hat{\rho} \hat{c}_{j}^{*}
                            -\frac{1}{2}\hat{c}_{j}^{*}\hat{c}_{j}\hat{\rho} 
                                -\frac{1}{2}\hat{\rho} \hat{c}_{j}^{*}\hat{c}_{j}\right),
\end{align}
where $\hat{\rho}$ is the density operator, $\hat{H}=\hat{H}^*$ 
represents the  system Hamiltonian, and $\hat{c}_{j}$ is a Lindblad operator 
that represents a system--reservoir interaction. 
Let 
$\hat{L}\triangleq\left[
\hat{c}_{1} \;\hat{c}_{2} \;\cdots \;\hat{c}_{K}
\right]^{\top}$ be the vector of Lindblad operators and for convenience, we call $\hat{L}$ the \emph{coupling vector}. 
The Lindblad 
master equation~\eqref{MME} describes the dynamics of a quantum system that interacts 
with many degrees of freedom in a dissipative environment. Under some conditions, the Lindblad 
master equation~\eqref{MME} can be strictly stable and has a unique steady state 
$\lim_{t \rightarrow \infty}\hat{\rho}(t)=\hat{\rho}(\infty)$. 
Based on this fact, it was shown in~\cite{KY12:pra} that 
any pure Gaussian state can be  prepared in the above dissipative 
way by selecting a suitable pair of operators $\left(\hat{H},\;\hat{L}\right)$. However, for some pure Gaussian states, the quantum systems that generate them are hard to implement experimentally, mainly because the operators $\left(\hat{H},\;\hat{L}\right)$ have a complex structure. 

This paper complements our previous work~\cite{MPW15:arxiv}. In this paper, we consider linear quantum systems subject to the following two constraints. (\textrm{i}) The Hamiltonian $\hat{H}$ is of the form 
$\hat{H}=\sum\limits_{j=1}^{N}\sum\limits_{k=j}^{N} g_{jk}\left(\hat{q}_{j}\hat{q}_{k}+\hat{p}_{j}\hat{p}_{k} \right)$, where $g_{jk}\in \mathbb{R}$, $ 1\le j\le k\le N$.  This type of Hamiltonian describes a set of passive beam-splitter-like interactions. (\textrm{ii})
 The system is \emph{locally} coupled to a single reservoir. In other words, the coupling vector $\hat{L}$ consists of only one Lindblad operator which acts only on a single site of the system. It is relatively simple to implement a quantum system subject to the two constraints (\textrm{i}) and (\textrm{ii}).   We parametrize the class of pure Gaussian states that can be prepared by this type of quantum system.

\textit{Notation.} 
$\mathbb{R}$ denotes the set of real numbers. $\mathbb{R}^{m \times n}$ denotes the set of real $m\times n$ matrices. $\mathbb{C}$ denotes the set of complex numbers. $\mathbb{C}^{m \times n}$ denotes the set of complex-entried
$m\times n$ matrices. $I_{n}$ denotes the $n\times n$ identity matrix. $0_{m\times n}$ denotes the $m\times n$ zero matrix.  The superscript ${}^{\ast}$ denotes either 
the complex conjugate of a complex number or the adjoint of an operator.
For a matrix $A=[A_{jk}]$ whose entries $A_{jk}$ are complex numbers or 
operators, we define $A^{\top}=[A_{kj}]$, $A^{\dagger}=[A_{kj}^{\ast}]$. For a matrix $A=A^{\top}\in \mathbb{R}^{n \times n}$, $A>0$ means that $A$ is positive definite. 
 $\diag[A_{1},\cdots,A_{n}]$ denotes a block diagonal matrix with diagonal blocks $A_{j}$, $j=1,2,\cdots,n$.  $\det(A)$ denotes the determinant of the matrix $A$.

\section{Preliminaries} \label{Preliminaries}
Consider an $N$-mode continuous-variable quantum system. Let $\hat{q}_{j}$ and $ \hat{p}_{j}$ be the position and momentum operators for the $j$th mode, respectively. Then they satisfy the following commutation relations ($\hbar=1$)
\begin{align*}
\left[\hat{q}_{j}, \hat{p}_{k}\right]=i\delta_{jk}, \quad \left[\hat{q}_{j}, \hat{q}_{k}\right]=0,\quad \text{and}\;\; \left[\hat{p}_{j}, \hat{p}_{k}\right]=0. 
\end{align*}
Let us define a column vector of operators
$\hat{x}=\left[\hat{q}_{1}\;\cdots\;\hat{q}_{N}\;\; \hat{p}_{1}\;\cdots\;\hat{p}_{N}\right]^{\top}$.  Then we have
\begin{align}
\label{commutation 1}
\left[\hat{x}, \hat{x}^{\top}\right]
=\hat{x}\hat{x}^{\top}-\left(\hat{x}\hat{x}^{\top}\right)^{\top}
  =i\Sigma, \quad \Sigma\triangleq\begin{bmatrix}
         0 & I_{N}\\
-I_{N} &0
\end{bmatrix}. 
\end{align}

Let $\hat{\rho}$ be the density operator of the system. 
Then the mean value of the  
vector $\hat{x} $ is given by 
$\langle \hat{x} \rangle 
=\left[\tr(\hat{q}_{1}\hat{\rho})\;\cdots\;\tr(\hat{q}_{N}\hat{\rho})
\; \tr(\hat{p}_{1}\hat{\rho})\;\cdots\;\tr(\hat{p}_{N}\hat{\rho})\right]^{\top}$ 
and the covariance matrix is given by   $V=\frac{1}{2}\langle \triangle\hat{x}{\triangle\hat{x}}^{\top}+(\triangle\hat{x}{\triangle\hat{x}}^{\top})^{\top} \rangle$, where $\triangle\hat{x}=\hat{x}-\langle \hat{x}\rangle$. 
A Gaussian state is entirely specified by its  mean vector $\langle \hat{x}\rangle$ 
and its covariance matrix $V$. Since the mean vector $\langle \hat{x}\rangle$ 
contains no information about noise and entanglement, 
we restrict our attention to zero-mean Gaussian states. The purity of a Gaussian state is defined by $p=\tr(\hat{\rho}^{2})=1/\sqrt{2^{2N}\det(V)}$. A Gaussian state with  covariance matrix $V$ is pure  if and only if $\det(V)=2^{-2N}$. 
In fact, when a Gaussian state is pure, its covariance matrix $V$   always has the following decomposition 
\begin{align}\label{covariance}
V=\frac{1}{2}\begin{bmatrix}
Y^{-1} &Y^{-1}X\\
XY^{-1} &XY^{-1}X+Y 
\end{bmatrix},
\end{align}
where $X=X^{\top}\in \mathbb{R}^{N \times N}$, $Y=Y^{\top}\in \mathbb{R}^{N \times N}$ and $Y>0$~\cite{MFL11:pra}. 
For the $N$-mode vacuum state, we have $X=0$ and $Y=I_{N}$. Let $Z\triangleq X+iY$. Given $Z$, a covariance matrix $V$ can be constructed from it using~\eqref{covariance}. Thus, the matrix $Z$ fully characterizes a pure Gaussian state. We refer to $Z=X+iY$ as a \emph{Gaussian graph matrix}~\cite{MFL11:pra}. Note that, a basic property of $Z$ is $\re(Z)=\re(Z)^{\top}$ and $\im(Z)=\im(Z)^{\top}>0$. This property guarantees physicality of the corresponding state.

Suppose that the system Hamiltonian in~\eqref{MME} is a quadratic function of $\hat{x}$, i.e., 
$\hat{H}=\frac{1}{2}\hat{x}^{\top}G\hat{x}$,
with $G=G^{\top}\in \mathbb{R}^{2N \times 2N}$, the coupling vector is a linear function of  $\hat{x}$, i.e., 
$\hat{L} = C \hat{x}$,
with $C\in \mathbb{C}^{K \times 2N}$, and the dynamics of the density operator $\hat{\rho}$ 
obey the Lindblad master equation~\eqref{MME}.   Then the corresponding dynamics of 
the mean vector $\langle \hat{x}(t) \rangle$ and the covariance matrix $V(t)$ can be described by 
\begin{numcases}{}
\frac{d\langle\hat{x}(t)\rangle}{dt}=\mathcal{A}\langle\hat{x}(t)\rangle, \label{meanfunction} \\
\frac{dV(t)}{dt}=\mathcal{A}V(t)+V(t)\mathcal{A}^{\top}+\mathcal{D},  \label{covfunction}
\end{numcases}
where $\mathcal{A}=\Sigma\left(G+\im(C^{\dagger}C)\right)$, 
$\mathcal{D}=\Sigma\re(C^{\dagger}C)\Sigma^{\top}$~\cite[Chapter 6]{WM10:book}. 
The linearity of the dynamics indicates that if the initial system state $\hat{\rho}(0)$ is Gaussian, 
then the system will always be Gaussian, with mean vector $\langle \hat{x}(t) \rangle$ 
and covariance matrix $V(t)$ evolving according to~\eqref{meanfunction} and~\eqref{covfunction}, respectively. 
We shall be particularly interested in the steady state of the system with the covariance matrix $V(\infty)$.
 Recently, a necessary and sufficient condition has been derived in~\cite{KY12:pra,Y12:ptrsa} for preparing a pure Gaussian steady state via reservoir engineering.  The result is summarized as follows. 
\begin{lem}[\cite{KY12:pra,Y12:ptrsa}]\label{lem0}
Let $Z=X+iY$ be the Gaussian graph matrix of an $N$-mode pure Gaussian state. Then this pure Gaussian state is  generated by the Lindblad 
master equation~\eqref{MME} if and only if
\begin{align} \label{G}
G=\begin{bmatrix}
XRX+YRY-\Gamma Y^{-1}X-XY^{-1}\Gamma^{\top} &-XR+\Gamma Y^{-1}\\
-RX+Y^{-1}\Gamma^{\top} &R
\end{bmatrix},
\end{align}
 and 
\begin{align} \label{C}
C=P^{\top}\left[-Z\;\; I_{N}\right], 
\end{align} 
where  $R=R^{\top}\in\mathbb{R}^{ N\times N}$,  $\Gamma=-\Gamma^{\top}\in\mathbb{R}^{ N\times N}$, and  $P\in \mathbb{C}^{N\times K}$ are free matrices satisfying the following rank condition
\begin{align}
\rank\left([P\;\;\;QP\;\;\;\cdots\;\;\;Q^{N-1}P]\right)=N,\;\;  Q\triangleq -iRY+Y^{-1}\Gamma. \label{rankconstraint}
\end{align}
\end{lem}
\begin{rmk}
A pair of matrices $\left(A_{1},\;A_{2}\right)$, where $A_{1}\in\mathbb{C}^{n\times n}$ and $A_{2}\in\mathbb{C}^{n\times m}$, is said to be \emph{controllable} if  
$\rank\left([A_{2}\;\;\;A_{1}A_{2}\;\;\;\cdots\;\;\;A_{1}^{n-1}A_{2}]\right)=n$~\cite[Theorem 6.1]{C99:book}.
Therefore, the rank constraint~\eqref{rankconstraint} is equivalent to saying that $\left(Q,\;P\right)$ is controllable. 
\end{rmk}

\section{Constraints} \label{Constraints}
According to Lemma~\ref{lem0}, for some pure Gaussian states, the quantum systems that generate them are hard to implement experimentally because the operators $\hat{H}=\frac{1}{2}\hat{x}^{\top}G\hat{x}$ and $\hat{L} = C \hat{x}$ have a complex structure. Here we discuss another route. We restrict our attention to  a class of linear quantum systems that are relatively simple to implement. Then we see which pure Gaussian states can be prepared by this type of system.
  We consider linear quantum systems subject to the following two constraints. 
\begin{enumerate}
\item The Hamiltonian $\hat{H}$ is of the form 
$\hat{H}=\sum\limits_{j=1}^{N}\sum\limits_{k=j}^{N} g_{jk}\left(\hat{q}_{j}\hat{q}_{k}+\hat{p}_{j}\hat{p}_{k} \right)$, where $g_{jk}\in \mathbb{R}$, $1\le j\le k\le N$. \label{constraint1}
\item The system is \emph{locally} coupled to a \emph{single} reservoir. That is, the coupling vector $\hat{L}$ is of the form $\hat{L}=c_{1}\hat{q}_{\ell}+c_{2}\hat{p}_{\ell}$, where $c_{1}\in \mathbb{C}$, $c_{2} \in \mathbb{C}$ and $\ell\in \{1,\; 2,\;\cdots,\; N\}$. \label{constraint2}
\end{enumerate}

\begin{rmk}
The Hamiltonian in~\ref{constraint1} can be rewritten in terms of the annihilation and creation operators as 
\begin{align*}
\hat{H}=\sum\limits_{j=1}^{N}\sum\limits_{k=j}^{N} g_{jk}\left(\hat{a}_{j}^{\ast}\hat{a}_{k}+\hat{a}_{j}\hat{a}_{k}^{\ast} \right), \quad g_{jk} \in \mathbb{R},\quad  1\le j\le k\le N, 
\end{align*}
where $\hat{a}_{j}=\frac{\hat{q}_{j}+i\hat{p}_{j}}{\sqrt{2}}$ and $\hat{a}_{j}^{\ast}=\frac{\hat{q}_{j}-i\hat{p}_{j}}{\sqrt{2}}$ are the annihilation and creation operators for the $j$th mode, respectively. A Hamiltonian $\hat{H}$ of this form is called \emph{passive} and describes beam-splitter-like interactions~\cite{SMD94:pra,GM16:pra}.  
\end{rmk}

\begin{rmk}
The constraint~\ref{constraint2} implies two crucial features of the system.  First, the system is coupled to  a \emph{single}  reservoir, i.e., $K=1$ in~\eqref{MME}. Second, the corresponding Lindblad operator acts only  on a \emph{single} mode of the system. 
\end{rmk}

To illustrate this type of linear quantum system, we consider  two examples. 
\begin{exam}~\label{exam1}
We consider a ring of three   quantum harmonic  oscillators ($N=3$), as depicted in Fig.~\ref{fig1}. The  quantum harmonic  oscillators are labelled $1$ to $3$. Suppose the Hamiltonian $\hat{H}$ is given by
\begin{align*}
\hat{H}&=\sum\limits_{j=1}^{3} g_{jj}\left(\hat{q}_{j}^{2}+\hat{p}_{j}^{2} \right)+g_{12}\left(\hat{q}_{1}\hat{q}_{2}+\hat{p}_{1}\hat{p}_{2} \right)\\
&+g_{13}\left(\hat{q}_{1}\hat{q}_{3}+\hat{p}_{1}\hat{p}_{3} \right)+g_{23}\left(\hat{q}_{2}\hat{q}_{3}+\hat{p}_{2}\hat{p}_{3} \right), 
\end{align*} 
where $g_{jk} \in \mathbb{R}$, $1\le j\le k\le 3$. In addition, the system-reservoir coupling vector $\hat{L}$ consists of only one Lindblad operator which acts on the third  mode of the system. It is given by
$\hat{L}= c_{1}\hat{q}_{3}+c_{2}\hat{p}_{3}$, where $c_{1}\in \mathbb{C}$ and $ c_{2} \in \mathbb{C}$. 
This linear quantum system satisfies the constraints~\ref{constraint1} and~\ref{constraint2}. 

\begin{figure}[htbp]
\begin{center}
\includegraphics[height=3cm]{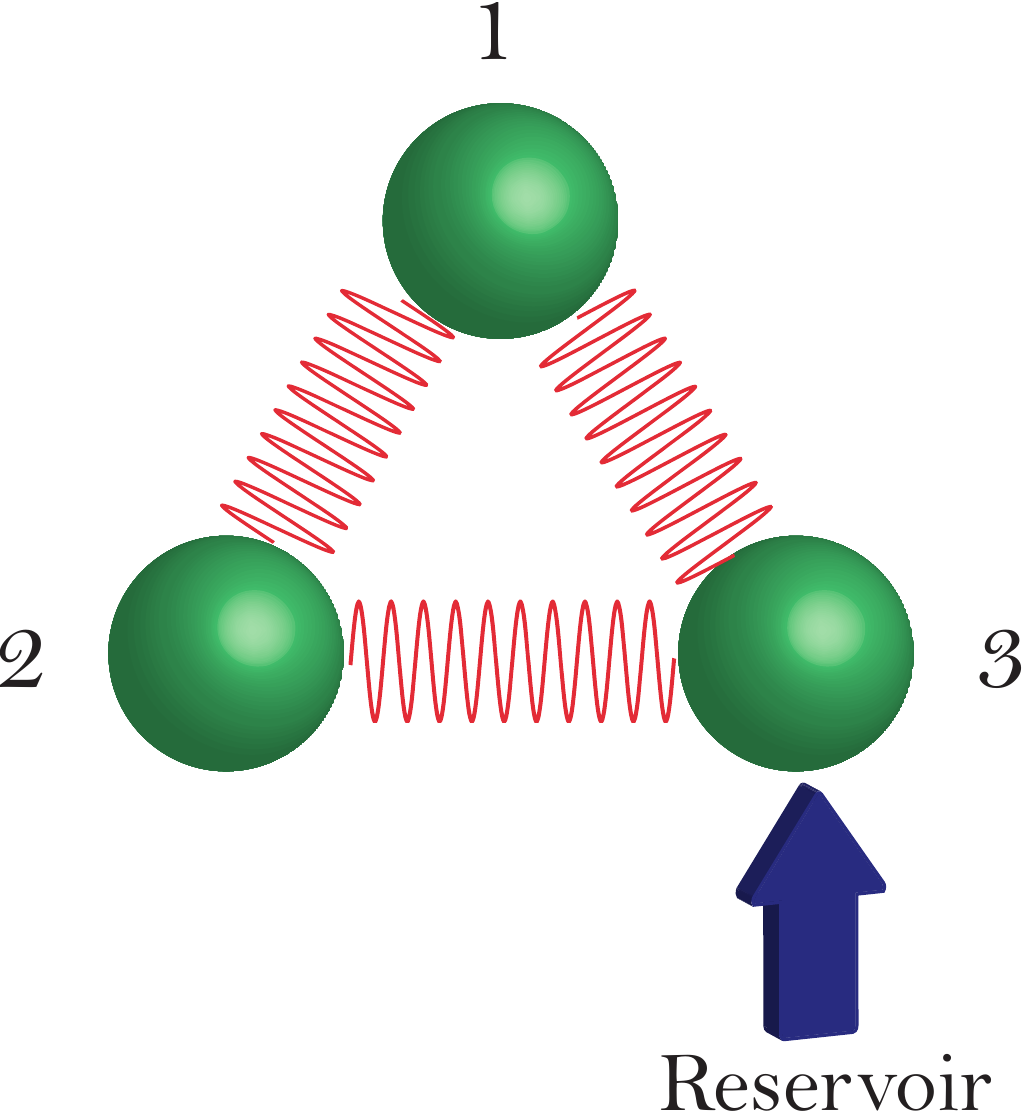}
\caption{
A ring of linearly coupled  quantum harmonic  oscillators. The system–-reservoir coupling vector $\hat{L}$ consists of only one Lindblad operator which acts on the third  mode of the system.}
\label{fig1}
\end{center}
\end{figure}
\end{exam}

\begin{exam}
We consider a chain of three  quantum harmonic  oscillators ($N=3$), as depicted in Fig.~\ref{fig2}. The  quantum harmonic  oscillators are labelled $1$ to $3$. Suppose the Hamiltonian $\hat{H}$ is given by
\begin{align*}
\hat{H}&=\sum\limits_{j=1}^{3} g_{jj}\left(\hat{q}_{j}^{2}+\hat{p}_{j}^{2} \right)+g_{12}\left(\hat{q}_{1}\hat{q}_{2}+\hat{p}_{1}\hat{p}_{2} \right)\\
&+g_{23}\left(\hat{q}_{2}\hat{q}_{3}+\hat{p}_{2}\hat{p}_{3} \right), 
\end{align*} 
where $g_{11},\; g_{22},\;g_{33},\;g_{12},\;g_{23}\in \mathbb{R}$. In addition, the system–-reservoir coupling vector $\hat{L}$ consists of only one Lindblad operator which acts  on the second mode of the system. It is given by
$
\hat{L}= c_{1}\hat{q}_{2}+c_{2}\hat{p}_{2}$,  where $c_{1}\in \mathbb{C}$ and $ c_{2} \in \mathbb{C}$. 
This linear quantum system satisfies the constraints~\ref{constraint1} and~\ref{constraint2}. 
\begin{figure}[htbp]
\begin{center}
\includegraphics[height=2cm]{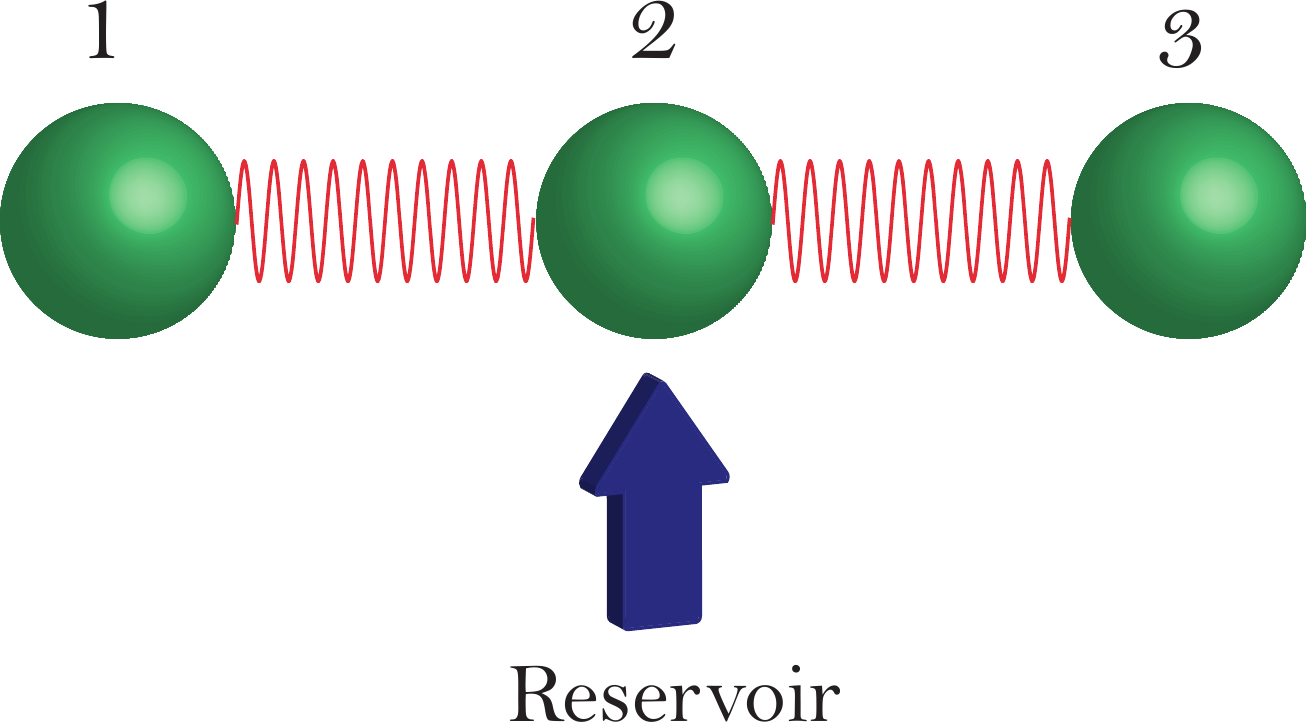}
\caption{
A chain of linearly coupled  quantum harmonic  oscillators. The system–-reservoir coupling vector $\hat{L}$ consists of only one Lindblad operator which acts on the second mode of the system.}
\label{fig2}
\end{center}
\end{figure}
\end{exam}

The class of linear quantum systems subject to~\ref{constraint1} and~\ref{constraint2} can be split into three disjoint subclasses (\uppercase\expandafter{\romannumeral1}, 
\uppercase\expandafter{\romannumeral2}, and 
\uppercase\expandafter{\romannumeral3}). The subclass \uppercase\expandafter{\romannumeral1} consists of all the \emph{unstable} linear quantum systems subject to the constraints~\ref{constraint1} and~\ref{constraint2}. For any system in the subclass \uppercase\expandafter{\romannumeral1}, there does not exist a  steady state. The subclass \uppercase\expandafter{\romannumeral2} consists of all the linear quantum systems with the constraints~\ref{constraint1} and~\ref{constraint2} that are strictly stable, and that evolve toward a mixed Gaussian steady state. The subclass \uppercase\expandafter{\romannumeral3} consists of all the linear quantum systems with the constraints~\ref{constraint1} and~\ref{constraint2}  that are strictly stable, and that evolve toward a pure Gaussian steady state. 

In the following section, we are particularly interested in the subclass \uppercase\expandafter{\romannumeral3}. Our objective is to characterize the class of pure Gaussian states that can be generated by linear quantum systems subject to~\ref{constraint1} and~\ref{constraint2}. In other words, we characterize all of the pure Gaussian states for which there exist a Hamiltonian $\hat{H}$ of the form~\ref{constraint1}  and a system--reservoir coupling vector $\hat{L}$ of the form~\ref{constraint2} such that the state is the unique steady state of the corresponding linear quantum system~\eqref{MME}. 

%

\section{Parametrization}

We give the main result. The following theorem provides a full mathematical parametrization of the pure Gaussian states that can be generated by systems subject to  the two constraints~\ref{constraint1} and~\ref{constraint2}. For clarity, we distinguish two cases: $N$ is odd and $N$ is even. 
\begin{thm} \label{theorem1}
$\;$
\begin{itemize}
\item An $N$-mode pure Gaussian state (where $N$ is odd) can be  generated by a  linear quantum system subject to the two constraints~\ref{constraint1} and~\ref{constraint2} if and only if its  Gaussian graph matrix $Z$ can be written as  
\begin{align}
Z&=\mathcal{P}^{\top}\begin{bmatrix}
\bar{z}  &0_{ 1 \times (N-1)}\\
0_{(N-1) \times 1} & \mathcal{Q}^{\top}\bar{Z}\mathcal{Q}
\end{bmatrix}\mathcal{P}, \label{thm11}\\
\bar{Z}&= \diag[\tilde{Z}_{1},\cdots,\tilde{Z}_{ \frac{N-1}{2} }],  \notag
\end{align}
 where $\mathcal{P}\in \mathbb{R}^{N \times N}$ is a permutation matrix, $\mathcal{Q}\in \mathbb{R}^{(N-1) \times (N-1)}$ is a real orthogonal matrix, 
$ \bar{z}\in \Lambda $,  and $\tilde{Z}_{j}\in \Delta$, $1\le j \le \frac{N-1}{2}$. 
\item An $N$-mode pure Gaussian state (where $N$ is even) can be  generated by a  linear quantum system subject to the two constraints~\ref{constraint1} and~\ref{constraint2} if and only if its  Gaussian graph matrix $Z$ can be written as  
\begin{align}
Z&=\mathcal{P}^{\top}\begin{bmatrix}
\bar{z}  &0_{ 1 \times (N-1)}\\
0_{(N-1) \times 1} & \mathcal{Q}^{\top}\bar{Z}\mathcal{Q}
\end{bmatrix}\mathcal{P}, \label{thm12}\\
 \bar{Z}&= \diag[\tilde{Z}_{1},\cdots,\tilde{Z}_{ \frac{N}{2} }], \notag
\end{align} 
 where $\mathcal{P}\in \mathbb{R}^{N \times N}$ is a permutation matrix, $\mathcal{Q}\in \mathbb{R}^{(N-1) \times (N-1)}$ is a real orthogonal matrix, 
$ \bar{z}\in \Lambda $, $\tilde{Z}_{1}=-\frac{1}{\bar{z}}$,  and $\tilde{Z}_{j}\in \Delta$, $2\le j \le \frac{N}{2}$. 
\end{itemize}
Here $
\Lambda \triangleq\{ z \;\;\big|\;\;  z \in \mathbb{C}\;\;\text{and}\;\;  \im(z) >0\}$, and $
\Delta \triangleq 
\bigg\{\mathpzc{Z}\;\;\Big|\; \mathpzc{Z}=\mathpzc{Z}^{\top}\in \mathbb{C}^{2\times 2}, \;\im\left(\mathpzc{Z}\right)>0, \;\big(\diag[1,-1]
\mathpzc{Z} \big)^{2} =-I_{2}, \; \text{and}\; \det(\mathpzc{Z}+ \frac{1}{\bar{z}}I_{2})=0\bigg\}$.
\end{thm}

The proof of Theorem~\ref{theorem1} is provided in the Appendix. We give some remarks to  explain Theorem~\ref{theorem1}.  

\begin{rmk}
Multiplying a Gaussian graph matrix $Z$ on the left by a permutation matrix $\mathcal{P}^{\top}$ and on the right by $\mathcal{P}$ simply corresponds to a relabeling of the $N$ modes.  Therefore, the role that the matrix $\mathcal{P}$ plays in  Equations~\eqref{thm11} and~\eqref{thm12} is not important from a practical point of view.  It simply amounts to a relabeling of the $N$ modes so that the first mode is the one that is coupled to the reservoir.  
\end{rmk}

\begin{rmk}
If $\bar{z}=i$ in Equations~\eqref{thm11} and~\eqref{thm12}, then from the condition that $\tilde{Z}_{j}\in \Delta$, we can deduce that $\tilde{Z}_{j}=i I_{2}$.  In this case, the whole state is a trivial pure Gaussian state, i.e., the $N$-mode vacuum state. It then follows from~\eqref{C} that the  system–-reservoir coupling vector $\hat{L}$ must be passive, i.e., $\hat{L}=c\left(\hat{q}_{\ell}+i\hat{p}_{\ell}\right)=\sqrt{2}c\hat{a}_{\ell}$, where $c\in\mathbb{C}$ and $\hat{a}_{\ell}=\frac{\hat{q}_{\ell}+i\hat{p}_{\ell}}{\sqrt{2}}$ is an annihilation operator. 
\end{rmk}

\begin{rmk}
The constraint $\im(z)>0$ in the definition of $\Lambda$ must be satisfied for the corresponding $Z$ to be a valid Gaussian graph matrix. For the same reason, the constraints $\mathpzc{Z}=\mathpzc{Z}^{\top}\in \mathbb{C}^{2\times 2}$ and $\im\left(\mathpzc{Z}\right)>0$ in the definition of $\Delta$ must be maintained. See the definition of the Gaussian graph matrix $Z$ in Section~\ref{Preliminaries} for details. 
\end{rmk}

\begin{rmk}
Theorem~\ref{theorem1} tells us that by choosing a permutation matrix $\mathcal{P}\in\mathbb{R}^{N\times N}$, a real orthogonal matrix $\mathcal{Q}\in\mathbb{R}^{(N-1)\times (N-1)}$, a complex number $\bar{z}\in \Lambda$ and several complex matrices $\tilde{Z}_{j}\in \Delta$, we can construct a pure Gaussian state that can be prepared by a system subject to the two constraints~\ref{constraint1} and~\ref{constraint2}. 
\end{rmk}

\begin{rmk}
As an extreme case, let us consider a one-mode pure Gaussian state ($N=1$). According to Theorem~\ref{theorem1}, a one-mode pure Gaussian state can be  generated by a  linear quantum system subject to the two constraints~\ref{constraint1} and~\ref{constraint2} if and only if its  Gaussian graph matrix $Z$ satisfies 
$Z=\mathcal{P}^{\top}
\bar{z}  
\mathcal{P}=\bar{z}$, where $ \bar{z}\in \Lambda $.  Since the set $\Lambda$ characterizes all of the one-mode pure Gaussian states, so we conclude that all of the one-mode pure Gaussian states can be generated by a linear quantum system subject to the two constraints~\ref{constraint1} and~\ref{constraint2}.  
\end{rmk}

\begin{rmk}
Let us consider the two-mode case ($N=2$). According to Theorem~\ref{theorem1}, a two-mode pure Gaussian state can be  generated by a  linear quantum system subject to the two constraints~\ref{constraint1} and~\ref{constraint2} if and only if its  Gaussian graph matrix $Z$ can be written as 
$Z=\mathcal{P}^{\top}\begin{bmatrix}
\bar{z} &0\\
0 &-\frac{1}{\bar{z}} \end{bmatrix}
\mathcal{P}$, where $ \bar{z}\in \Lambda $.  Since the role of the permutation matrix $\mathcal{P}$ is not important here, it follows that the Gaussian graph matrix of the two-mode state is of the form $Z=\begin{bmatrix}
\bar{z} &0\\
0 &-\frac{1}{\bar{z}} \end{bmatrix}$, where $ \bar{z}\in \Lambda $.  The structure of the system that generates this state is shown in Fig.~\ref{figrmk1}. It follows from the form of the matrix $Z$ that the two modes are not entangled~\cite{MWPY16:arxiv}. 

\newsavebox{\smlmat}
\savebox{\smlmat}{$Z=\begin{bmatrix}\bar{z} &0 \\
0 &-\frac{1}{\bar{z}} \end{bmatrix}$}
\begin{figure}[htbp]
\begin{center}
\includegraphics[height=2.2cm]{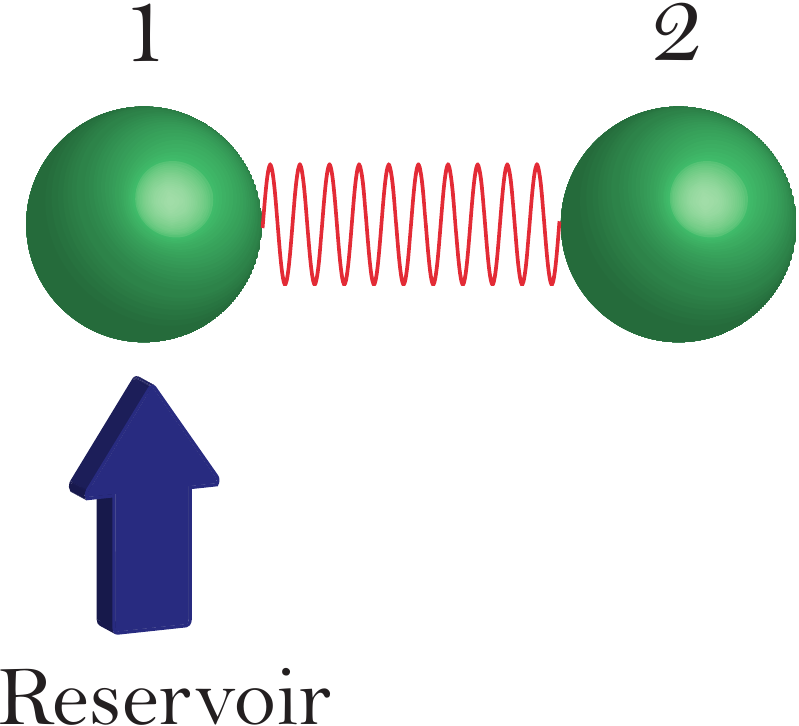}
\caption{
The structure of a quantum system that generates the two-mode pure Gaussian state with  Gaussian graph matrix~\usebox{\smlmat}.}
\label{figrmk1}
\end{center}
\end{figure} 
\end{rmk}

\begin{rmk}
Let us consider the three-mode case ($N=3$). According to Theorem~\ref{theorem1}, a three-mode pure Gaussian state can be  generated by a  linear quantum system subject to the two constraints~\ref{constraint1} and~\ref{constraint2} if and only if its  Gaussian graph matrix  
$Z$ can be written as $
Z=\mathcal{P}^{\top}\begin{bmatrix}
\bar{z}  &0_{ 1 \times 2}\\
0_{2 \times 1} & \mathcal{Q}^{\top}\tilde{Z}_{1}\mathcal{Q}
\end{bmatrix}\mathcal{P}$, 
 where   
$ \bar{z}\in \Lambda $,  and $\tilde{Z}_{1}\in \Delta$. As mentioned, the constraint defining $\Lambda$, i.e., $\im(z)>0$, comes from the definition of a Gaussian graph matrix. Similar constraints (i.e., $\left\{\mathpzc{Z}=\mathpzc{Z}^{\top}\in \mathbb{C}^{2\times 2}\;\;\text{and}\; \im\left(\mathpzc{Z}\right)>0\right\}$) apply to the set $\Delta$. Let us analyze the other two constraints in $\Delta$. One is $\big(\diag[1,-1]
\mathpzc{Z}\big)^{2} =-I_{2}$, i.e., $\mathpzc{Z}\diag[1,-1]\mathpzc{Z}=-\diag[1,-1]$, which is similar to what we have obtained in Theorem~1 of~\cite{MPW15:arxiv}. This should not be surprising, since  we can think of the mode coupled to the reservoir as an auxiliary mode. Then the other two modes are coupled to a single  reservoir. This is indeed the case considered in~\cite{MPW15:arxiv}. The other constraint defining $\Delta$, i.e., $\det(\mathpzc{Z}+ \frac{1}{\bar{z}}I_{2})=0$, is equivalent to saying that $-\frac{1}{\bar{z}}$ is an eigenvalue of the matrix $\tilde{Z}_{j}$. This constraint comes from the \emph{locality} requirement on the original system--reservoir coupling. 
\end{rmk}

\begin{lem}\label{theorem2} 
Given $ \bar{z}\in \Lambda $, then the set $\Delta$ contains only two matrices, i.e., 
\begin{align*}
\Delta =\left\{\begin{bmatrix}\frac{\bar{z}^{2}-1}{2\bar{z}} &\frac{\bar{z}^{2}+1}{2\bar{z}}\\ \frac{\bar{z}^{2}+1}{2\bar{z}} &\frac{\bar{z}^{2}-1}{2\bar{z}} \end{bmatrix},\quad\begin{bmatrix}\frac{\bar{z}^{2}-1}{2\bar{z}} &-\frac{\bar{z}^{2}+1}{2\bar{z}}\\ -\frac{\bar{z}^{2}+1}{2\bar{z}} &\frac{\bar{z}^{2}-1}{2\bar{z}} \end{bmatrix} \right\}. 
\end{align*}
\end{lem}

The proof of Lemma~\ref{theorem2} is provided in the Appendix. Using Lemma~\ref{theorem2}, we have an equivalent version of Theorem~\ref{theorem1}.  

\begin{thm} \label{theorem3}
$\;$
\begin{itemize}
\item An $N$-mode pure Gaussian state (where $N$ is odd) can be  generated by a  linear quantum system subject to the two constraints~\ref{constraint1} and~\ref{constraint2} if and only if its  Gaussian graph matrix $Z$ can be written as 
\begin{align}
Z&=\mathcal{P}^{\top}\begin{bmatrix}
\bar{z}  &0_{ 1 \times (N-1)}\\
0_{(N-1) \times 1} & \mathcal{Q}^{\top}\bar{Z}\mathcal{Q}
\end{bmatrix}\mathcal{P}, \label{thm31}\\
\bar{Z}&= \diag[\tilde{Z}_{1},\cdots,\tilde{Z}_{ \frac{N-1}{2} }],  \notag
\end{align}
 where $\mathcal{P}\in \mathbb{R}^{N \times N}$ is a permutation matrix, $\mathcal{Q}\in \mathbb{R}^{(N-1) \times (N-1)}$ is a real orthogonal matrix, 
$ \bar{z}\in \Lambda $,  and $\tilde{Z}_{j}=\begin{bmatrix}\frac{\bar{z}^{2}-1}{2\bar{z}} &\frac{\bar{z}^{2}+1}{2\bar{z}}\\ \frac{\bar{z}^{2}+1}{2\bar{z}} &\frac{\bar{z}^{2}-1}{2\bar{z}} \end{bmatrix}\;\text{or}\;\begin{bmatrix}\frac{\bar{z}^{2}-1}{2\bar{z}} &-\frac{\bar{z}^{2}+1}{2\bar{z}}\\ -\frac{\bar{z}^{2}+1}{2\bar{z}} &\frac{\bar{z}^{2}-1}{2\bar{z}} \end{bmatrix} $, $1\le j \le \frac{N-1}{2}$. 
\item An $N$-mode pure Gaussian state (where $N$ is even) can be  generated by a  linear quantum system subject to the two constraints~\ref{constraint1} and~\ref{constraint2} if and only if its  Gaussian graph matrix $Z$ can be written as 
\begin{align}
Z&=\mathcal{P}^{\top}\begin{bmatrix}
\bar{z}  &0_{ 1 \times (N-1)}\\
0_{(N-1) \times 1} & \mathcal{Q}^{\top}\bar{Z}\mathcal{Q}
\end{bmatrix}\mathcal{P}, \label{thm32}\\
 \bar{Z}&= \diag[\tilde{Z}_{1},\cdots,\tilde{Z}_{ \frac{N}{2} }], \notag
\end{align} 
 where $\mathcal{P}\in \mathbb{R}^{N \times N}$ is a permutation matrix, $\mathcal{Q}\in \mathbb{R}^{(N-1) \times (N-1)}$ is a real orthogonal matrix, 
$ \bar{z}\in \Lambda $, $\tilde{Z}_{1}=-\frac{1}{\bar{z}}$,  and $\tilde{Z}_{j}=\begin{bmatrix}\frac{\bar{z}^{2}-1}{2\bar{z}} &\frac{\bar{z}^{2}+1}{2\bar{z}}\\ \frac{\bar{z}^{2}+1}{2\bar{z}} &\frac{\bar{z}^{2}-1}{2\bar{z}} \end{bmatrix}\;\text{or}\;\begin{bmatrix}\frac{\bar{z}^{2}-1}{2\bar{z}} &-\frac{\bar{z}^{2}+1}{2\bar{z}}\\ -\frac{\bar{z}^{2}+1}{2\bar{z}} &\frac{\bar{z}^{2}-1}{2\bar{z}} \end{bmatrix} $, $2\le j \le \frac{N}{2}$. 
\end{itemize}
\end{thm}

\begin{rmk}
Once we obtain a pure Gaussian state using Theorem~\ref{theorem3}, we can immediately construct a linear quantum system that generates such a state and also satisfies the constraints~\ref{constraint1} and~\ref{constraint2}. The construction method is given  as follows. 

\textbf{Case 1} ($N$ is odd). If $\tilde{Z}_{j}=\begin{bmatrix}\frac{\bar{z}^{2}-1}{2\bar{z}} &\frac{\bar{z}^{2}+1}{2\bar{z}}\\ \frac{\bar{z}^{2}+1}{2\bar{z}} &\frac{\bar{z}^{2}-1}{2\bar{z}} \end{bmatrix}$, choose $\tilde{R}_{21,j}=\begin{bmatrix} \bar{\tau}_{j} \\ -\bar{\tau}_{j}\end{bmatrix}$, where $\bar{\tau}_{j}\in \mathbb{R}$ and $\bar{\tau}_{j} \ne 0$.  Otherwise, choose $\tilde{R}_{21,j}=\begin{bmatrix} \bar{\tau}_{j} \\ \bar{\tau}_{j}\end{bmatrix}$, where $\bar{\tau}_{j}\in \mathbb{R}$ and $\bar{\tau}_{j} \ne 0$.  Let $
 \bar{R}_{21}=\begin{bmatrix}\tilde{R}_{21,1}^{\top} \;\cdots\;\tilde{R}_{21,  \frac{(N-1)}{2}}^{\top}\end{bmatrix}^{\top}$.  Let $
 \bar{R}_{22}=\diag\left[
r_{1},\;\; 
-r_{1},\;\; 
r_{2},\; \;
-r_{2},\;\; 
\cdots ,\;\; 
r_{\frac{N-1}{2}},\; \;
-r_{\frac{N-1}{2}}\right]$, where $r_{j}\in \mathbb{R}$, $r_{j} \ne 0$ and $|r_{j}|\ne |r_{k}|$  whenever $ j\ne k$.  In Lemma~\ref{lem0}, let us choose $R=\mathcal{P}^{\top}\begin{bmatrix}
0  &\bar{R}_{21}^{\top}\mathcal{Q}\\
\mathcal{Q}^{\top}\bar{R}_{21} & \mathcal{Q}^{\top}\bar{R}_{22}\mathcal{Q}
\end{bmatrix}\mathcal{P}$, $\Gamma=XRY$ and $P=\mathcal{P}^{\top}\begin{bmatrix}\tau_{p} &0_{1\times (N-1)}\end{bmatrix}^{\top}$, where $\tau_{p}\in \mathbb{C}$ and $\tau_{p}\ne 0$.
Then calculate the matrices $G$ and $C$ using~\eqref{G} and~\eqref{C}, respectively. The resulting linear quantum system with Hamiltonian $\hat{H}=\frac{1}{2}\hat{x}^{\top}G\hat{x}$ and coupling vector $\hat{L}=C\hat{x}$ is strictly stable and   generates the pure Gaussian state. It can be shown that this quantum system also satisfies the two constraints~\ref{constraint1} and~\ref{constraint2}. 

\textbf{Case 2} ($N$ is even). Choose $\tilde{R}_{21,1}=\bar{\tau}_{1} $, where $\bar{\tau}_{1}\in \mathbb{R}$ and $\bar{\tau}_{1} \ne 0$.  If $\tilde{Z}_{j}=\begin{bmatrix}\frac{\bar{z}^{2}-1}{2\bar{z}} &\frac{\bar{z}^{2}+1}{2\bar{z}}\\ \frac{\bar{z}^{2}+1}{2\bar{z}} &\frac{\bar{z}^{2}-1}{2\bar{z}} \end{bmatrix}$, $j\ge 2$, choose $\tilde{R}_{21,j}=\begin{bmatrix} \bar{\tau}_{j} \\ -\bar{\tau}_{j}\end{bmatrix}$, where $\bar{\tau}_{j}\in \mathbb{R}$ and $\bar{\tau}_{j} \ne 0$.  Otherwise, choose $\tilde{R}_{21,j}=\begin{bmatrix} \bar{\tau}_{j} \\ \bar{\tau}_{j}\end{bmatrix}$, where $\bar{\tau}_{j}\in \mathbb{R}$ and $\bar{\tau}_{j} \ne 0$.  Let $
 \bar{R}_{21}=\begin{bmatrix}\tilde{R}_{21,1}^{\top} \;\cdots\;\tilde{R}_{21,  \frac{N}{2}}^{\top}\end{bmatrix}^{\top}$.  Let $
 \bar{R}_{22}=\diag\left[0,\;\;
r_{2},\; \;
-r_{2},\;\; 
\cdots ,\;\; 
r_{\frac{N}{2}},\; \;
-r_{\frac{N}{2}}\right]$, where $r_{j}\in \mathbb{R}$, $r_{j} \ne 0$, $j\ge 2$,  and $|r_{j}|\ne |r_{k}|$  whenever $ j\ne k$.  In Lemma~\ref{lem0}, let us choose $R=\mathcal{P}^{\top}\begin{bmatrix}
0  &\bar{R}_{21}^{\top}\mathcal{Q}\\
\mathcal{Q}^{\top}\bar{R}_{21} & \mathcal{Q}^{\top}\bar{R}_{22}\mathcal{Q}
\end{bmatrix}\mathcal{P}$, $\Gamma=XRY$ and $P=\mathcal{P}^{\top}\begin{bmatrix}\tau_{p} &0_{1\times (N-1)}\end{bmatrix}^{\top}$, where $\tau_{p}\in \mathbb{C}$ and $\tau_{p}\ne 0$.
Then calculate the matrices $G$ and $C$ using~\eqref{G} and~\eqref{C}, respectively. The resulting linear quantum system with Hamiltonian $\hat{H}=\frac{1}{2}\hat{x}^{\top}G\hat{x}$ and coupling vector $\hat{L}=C\hat{x}$ is strictly stable and   generates the pure Gaussian state. It can be shown that  this system also satisfies the two constraints~\ref{constraint1} and~\ref{constraint2}. 
\end{rmk}

\section{Example}
We consider the five-mode pure Gaussian state generated in~\cite{ZLV15:pra}.  The Gaussian graph matrix of this pure Gaussian state is given by 
\scriptsize
\begin{align}
&Z=   \notag \\
&          i\begin{bmatrix}
\cosh(2\alpha)         &0  &0       &0         & \sinh(2\alpha)                  \\
 0               &\cosh(2\alpha) &0         & -  \sinh(2\alpha)     & 0                   \\
  0              &0             &\cosh(2\alpha)+\sinh(2\alpha) & 0             & 0                 \\
   0            & -  \sinh(2\alpha) & 0             &\cosh(2\alpha)        & 0       \\
 \sinh(2\alpha)   &0 & 0     &0       &\cosh(2\alpha)
\end{bmatrix}, \label{example1}
\end{align} 
\normalsize
where $\alpha\in \mathbb{R}$. Let us choose a permutation matrix 
$
\mathcal{P}=\begin{bmatrix}
0 &0 &1 &0 &0 \\
1 &0 &0 &0 &0\\
0 &0 &0 &0 &1 \\
0 &1 &0 &0 &0 \\
0 &0 &0 &1 &0 
\end{bmatrix}
$ and a real orthogonal matrix $\mathcal{Q}=\frac{\sqrt{2}}{2}\begin{bmatrix}
-1 &0 &-1 &0 \\
0 &-1 &0 &1  \\
0 &1 &0 &1  \\
-1 &0 &1 &0  
\end{bmatrix}$. Then we have 
$
Z= \mathcal{P}^{\top}\begin{bmatrix}
\bar{z}  &0_{ 1 \times 4}\\
0_{4 \times 1} & \mathcal{Q}^{\top}\bar{Z}\mathcal{Q}
\end{bmatrix}\mathcal{P}$, $\bar{Z}= \diag[\tilde{Z}_{1},\tilde{Z}_{2}]$, 
where $ \bar{z}=i \left(\cosh(2\alpha)+\sinh(2\alpha)\right)$,  and $\tilde{Z}_{j}=i\begin{bmatrix}
\cosh(2\alpha) & (-1)^{j-1} \sinh(2\alpha)\\
(-1)^{j-1}\sinh(2\alpha)        &\cosh(2\alpha)
\end{bmatrix}$,  $j=1,\; 2$. It can be verified that $\bar{z}\in \Lambda$ and $\tilde{Z}_{j}\in \Delta$, $j=1,\; 2$.  According to Theorem~\ref{theorem1}, this pure Gaussian state can be generated by a linear quantum system satisfying the two constraints~\ref{constraint1} and~\ref{constraint2}. Next, we construct such a system. Let $\bar{R}_{21}= \sqrt{2}\begin{bmatrix}
-1 & 1   &1 &1
\end{bmatrix}^{\top}$ and $\bar{R}_{22}=\diag[
1, \;\;-1, \;\;3, \;\;-3] $. Then in Lemma~\ref{lem0}, let us choose $
R =\mathcal{P}^{\top}\begin{bmatrix}
0 &\bar{R}_{21}^{\top}\mathcal{Q}\\
\mathcal{Q}^{\top}\bar{R}_{21} &\mathcal{Q}^{\top}\bar{R}_{22}\mathcal{Q}
\end{bmatrix}\mathcal{P} 
 =\begin{bmatrix}
   -1   & 2        & 0       &  0     &    0\\
   2    &-1    &2        &  0      &   0\\
         0    &2         & 0   & 2    &     0\\
         0         &0    &2    & 1   & 2 \\
         0         &0         &0   & 2    &1 
\end{bmatrix}$, $\Gamma =0_{5\times 5}$, and $
P=i\frac{\cosh(\alpha)-\sinh(\alpha)}{\sqrt{2}}\begin{bmatrix}0 &0 &1 &0 &0 \end{bmatrix}^{\top}$. 
It can be verified that $YRY=R$ and 
\begin{align*}
&\rank\left([P\;\;\;QP\;\;\;Q^{2}P\;\;\;Q^{3}P\;\;\;Q^{4}P]\right)\\
=& \rank\left([P\;\;\;-iRYP\;\;\;-R^{2}P\;\;\;iR^{3}YP\;\;\;R^{4}P]\right)\\
=&5.
\end{align*}
Therefore, the resulting linear quantum system is strictly stable and  generates the target pure Gaussian state given in~\eqref{example1}. Substituting $R$, $\Gamma$ and $P$ into~\eqref{G} and~\eqref{C}, we obtain the Hamiltonian of the system
\begin{align*}
\hat{H}&= -\frac{1}{2}\left(\hat{q}_{1}^{2}+\hat{p}_{1}^{2} +\hat{q}_{2}^{2}+\hat{p}_{2}^{2} \right)+ \frac{1}{2}\left(\hat{q}_{4}^{2}+\hat{p}_{4}^{2} +\hat{q}_{5}^{2}+\hat{p}_{5}^{2} \right) \\
&+2\left(\hat{q}_{1}\hat{q}_{2}+\hat{p}_{1} \hat{p}_{2} + \hat{q}_{2}\hat{q}_{3}+\hat{p}_{2}\hat{p}_{3}\right)\\
& + 2\left(\hat{q}_{3}\hat{q}_{4}+\hat{p}_{3}\hat{p}_{4}  +\hat{q}_{4}\hat{q}_{5}+\hat{p}_{4}\hat{p}_{5} \right), 
\end{align*} 
and the coupling vector 
\begin{align*}
\hat{L}&= i\frac{\cosh(\alpha)-\sinh(\alpha)}{\sqrt{2}}\left[-i(\left(\cosh(2\alpha)+\sinh(2\alpha)\right)\hat{q}_{3}+\hat{p}_{3}\right]\\
&= \frac{\cosh(\alpha)+\sinh(\alpha)}{\sqrt{2}} \hat{q}_{3}+i\frac{\cosh(\alpha)-\sinh(\alpha)}{\sqrt{2}} \hat{p}_{3}  \\
&=\cosh(\alpha)\hat{a}_{3}+\sinh(\alpha)\hat{a}_{3}^{\ast}.  
\end{align*} 
The coupling operator $\hat{L}$ represents a standard dissipative reservoir that acts only on the third mode. The eigenstate corresponding to the zero eigenvalue of $\hat{L}$ is a squeezed state. The third system mode is then prepared in a squeezed state, while other modes are coupled via passive interactions. This leads to entanglement across the system, in a similar way to which the interference of a squeezed optical mode with a vacuum at a beam splitter results in entangled output modes~\cite{KSBK02:pra}.  
The structure of the system is shown in Fig.~\ref{fig4}. It is a chain of  quantum harmonic  oscillators with nearest--neighbour Hamiltonian interactions. Only the central oscillator is coupled to the reservoir. 
\begin{figure}[htbp]
\begin{center}
\includegraphics[height=2cm]{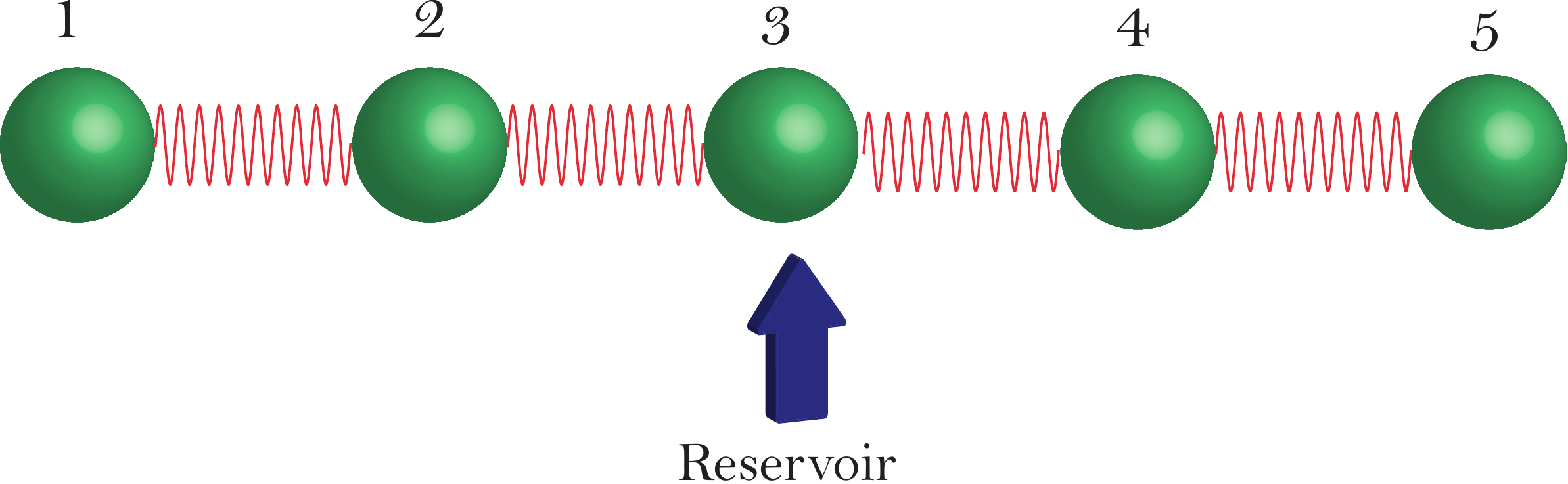}
\caption{
The  pure Gaussian state given by~\eqref{example1} can be generated in a chain of linearly coupled  quantum harmonic  oscillators with nearest--neighbour Hamiltonian interactions. Only the central oscillator is coupled to the reservoir. }
\label{fig4}
\end{center}
\end{figure}

The oscillators of the system are entangled in pairs. The first oscillator is entangled with the fifth oscillator and the second oscillator is entangled with the fourth oscillator.  The central  oscillator is not entangled with the other oscillators. The amount of entanglement can be quantified using the logarithmic negativity $\mathcal{E}$~\cite{VW02:pra,P05:prl,ASI04:pra}. The values are given by $\mathcal{E}_{(1,5)}=\mathcal{E}_{(2,4)} =2|\alpha|$. For a more detailed discussion of this example, we refer the reader to~\cite{ZLV15:pra}. 

\section{Conclusion} 
In this paper, we consider linear quantum systems subject to constraints. First, we assume that the Hamiltonian $\hat{H}$ is of the form 
$\hat{H}=\sum\limits_{j=1}^{N}\sum\limits_{k=j}^{N} g_{jk}\left(\hat{q}_{j}\hat{q}_{k}+\hat{p}_{j}\hat{p}_{k} \right)$, where $g_{jk}\in \mathbb{R}$, $ 1\le j\le k\le N$.  Second, we assume that the system is locally  coupled to a single reservoir. Then we give a full mathematical parametrization of the pure Gaussian states that can be prepared using this type of quantum system.

\section{Appendix}
In this section, we provide the proof of Theorem~\ref{theorem1}. First, we provide some preliminary results which will be used in the proof of Theorem~\ref{theorem1}.
\begin{lem}[\cite{MWPY16:arxiv}]\label{lem1}
Suppose that an $N$-mode pure Gaussian state is generated in a linear quantum system with a single dissipative process and that the corresponding Lindblad operator  acts only on the $\ell$th mode of the system, then 
\begin{align*}
 Z_{(\ell,j)}=Z_{(j,\ell)}=0,\quad \forall j\ne \ell, 
\end{align*}
where $Z_{(\ell,j)}$ denotes the $(\ell,j)$ element of the Gaussian graph matrix $Z$. In addition, the $\ell$th mode is not entangled with the other modes. 
\end{lem}

\begin{lem}[\cite{C99:book}] \label{lem3}
If a pair of matrices $\left(A_{1},\;A_{2}\right)$ is controllable, where $A_{1}\in\mathbb{C}^{n\times n}$ and $A_{2}\in\mathbb{C}^{n\times m}$, then $\left(F^{-1}A_{1}F,\;F^{-1}A_{2}\right)$ is controllable  where $F\in\mathbb{C}^{n\times n}$ is a non-singular matrix.
\end{lem}

\begin{lem} \label{lem2}
Suppose $A=\begin{bmatrix}
A_{11}   &A_{12}\\
 A_{21}   &A_{22}
\end{bmatrix}$, where $A_{11}\in \mathbb{C}$, $A_{12}\in \mathbb{C}^{1  \times (n-1)}$, $A_{21}\in \mathbb{C}^{ (n-1)\times 1 }$, and $A_{22}\in \mathbb{C}^{(n-1) \times (n-1)}$, and $\xi=\begin{bmatrix}\tau &0_{1\times(n-1)}\end{bmatrix}^{\top}\in  \mathbb{C}^{n \times 1}$, where $\tau \in  \mathbb{C}$ and $\tau\ne 0$. Then the pair $(A,\; \xi)$ is controllable if and only if the pair $(A_{22},\; A_{21})$ is controllable. 
\end{lem}

\begin{proof}
If $(A,\; \xi)$ is controllable, then  the matrix $\left[A-\lambda I_{n}\;\;  \xi\right]$ has full row rank for all $\lambda \in  \mathbb{C}$~\cite[Theorem 3.1]{ZJK96:book}. That is,  $\rank\left(\begin{bmatrix}
A_{11}-\lambda    &A_{12} &\tau \\
 A_{21}   &A_{22}-\lambda I_{n-1}  &0_{(n-1) \times 1}
\end{bmatrix}\right)=n$ for all $\lambda \in  \mathbb{C}$. It follows that $\rank\left(\left[ A_{21}\;\; A_{22}-\lambda I_{n-1}\right]\right)=n-1$ for all $\lambda \in  \mathbb{C}$. Therefore, $(A_{22},\; A_{21})$ is controllable~\cite[Theorem 3.1]{ZJK96:book}.  Conversely, if $(A_{22},\; A_{21})$ is controllable, then $\rank\left(\left[A_{22}-\lambda I_{n-1}\;\; A_{21}\right]\right)=n-1$ for all $\lambda \in  \mathbb{C}$. Since $\tau\ne 0$, it follows that $\rank\left(\begin{bmatrix}
A_{11}-\lambda    &A_{12}  &\tau\\
 A_{21}   &A_{22}-\lambda I_{n-1}   &0_{(n-1) \times 1}
\end{bmatrix}\right)=n$ for all $\lambda \in  \mathbb{C}$. That is,  the matrix $\left[A-\lambda I_{n}\;\;  \xi\right]$ has full row rank for all $\lambda \in  \mathbb{C}$. Therefore, $(A,\; \xi)$ is controllable. 
\end{proof}

\begin{lem} \label{lem4}
Suppose $A=\diag\left[A_{11},\;A_{22},\;\cdots, \; A_{mm}\right]$, where $A_{jj}\in \mathbb{C}^{n_{j}\times n_{j}}$, $1\le j\le m$, and $\xi=\begin{bmatrix}\xi_{1}^{\top} & \xi_{2}^{\top} &\cdots & \xi_{m}^{\top} \end{bmatrix}^{\top}$, where $\xi_{j}\in  \mathbb{C}^{n_{j} \times 1}$, $1\le j\le m$.  Then the pair $(A,\; \xi)$ is controllable if and only if all the pairs $(A_{jj},\; \xi_{j})$, $1\le j\le m$, are controllable and $A_{jj}$ and $A_{kk}$, $j\ne k$, have no common eigenvalues. 
\end{lem}

\begin{proof}
\emph{Necessity}. If the pair $(A,\; \xi)$ is controllable, then the matrix $\left[A-\lambda I_{\sum_{j=1}^{m} n_{j}}\;\;\;  \xi\right]$ has full row rank for all $\lambda \in  \mathbb{C}$~\cite[Theorem 3.1]{ZJK96:book}. That is, 
\small
\begin{align}
&\rank\left(\begin{bmatrix}
A_{11}-\lambda I_{n_{1}}   & & & &\xi_{1}\\
    &A_{22}-\lambda I_{n_{2}} &&  &\xi_{2}\\
    & &\ddots  & &\vdots\\
  & &  &A_{mm}-\lambda I_{n_{m}}  &\xi_{m}  
\end{bmatrix}\right) \notag\\
=&\sum_{j=1}^{m} n_{j}. \label{lem4proof}
\end{align}
\normalsize
Then it follows that the matrix $\left[A_{jj}-\lambda I_{n_{j}}\;\;  \xi_{j}\right]$ has full row rank for all $\lambda \in  \mathbb{C}$. Therefore, all the pairs $(A_{jj},\; \xi_{j})$, $1\le j\le m$, are controllable~\cite[Theorem 3.1]{ZJK96:book}. 

To prove the second part of necessity, without loss of generality, we assume that  $A_{11}$ and $A_{22}$ share a common eigenvalue $\lambda $. Then the matrix $\begin{bmatrix}
A_{11} &0 \\
 0   &A_{22} 
\end{bmatrix} $ is  a derogatory matrix~\cite{HJ12:book,B09:book}. Using Lemma~6 in~\cite{MPW15:arxiv}, the pair $\left(\begin{bmatrix}
A_{11} &0 \\
0    &A_{22} 
\end{bmatrix},  \begin{bmatrix}
\xi_{1}\\
\xi_{2}
\end{bmatrix} \right)$ cannot be controllable. But we already know from the previous result that if the pair $(A,\; \xi)$ is controllable then the pair $\left(\begin{bmatrix}
A_{11} &0 \\
 0   &A_{22} 
\end{bmatrix},  \begin{bmatrix}
\xi_{1}\\
\xi_{2}
\end{bmatrix} \right)$ must be controllable. Therefore, we
reach a contradiction. We conclude that  $A_{jj}$ and $A_{kk}$, $j\ne k$, have no common eigenvalues. 

\emph{Sufficiency}. To show $(A,\; \xi)$ is controllable, we need to prove that the matrix $\left[A-\lambda I_{\sum_{j=1}^{m} n_{j}}\;\;  \xi\right]$ has full row rank for all $\lambda \in  \mathbb{C}$. That is, we need to prove that the matrix $\left[A-\lambda I_{\sum_{j=1}^{m} n_{j}}\;\;  \xi\right]$ has full row rank for any eigenvalue $\lambda$ of $A$. Now suppose $\lambda$ is an arbitrary eigenvalue of $A$. Since $A_{jj}$ and $A_{kk}$,  $j\ne k$, have no common eigenvalues, it follows that $\lambda$ is an eigenvalue of only one block. Without loss of generality, we assume that $\lambda$ is an eigenvalue of $A_{11}$. Then we have  $\det(A_{11}-\lambda I_{n_{1}})=0$ and $\det(A_{jj}-\lambda I_{n_{j}})\ne 0$, $2\le j\le m$. Since  $(A_{11},\; \xi_{1})$ is controllable, it follows that $\left[A_{11}-\lambda I_{n_{1}}\;\;  \xi_{1}\right]$ has full row rank. Then we have 
$ \begin{bmatrix}
A_{11}-\lambda I_{n_{1}}   & & & &\xi_{1}\\
    &A_{22}-\lambda I_{n_{2}} &&  &0\\
    & &\ddots  & &\vdots\\
  & &  &A_{mm}-\lambda I_{n_{m}}  &0  
\end{bmatrix} $ has full row rank. As a consequence, it can be shown that $ \begin{bmatrix}
A_{11}-\lambda I_{n_{1}}   & & & &\xi_{1}\\
    &A_{22}-\lambda I_{n_{2}} &&  &\xi_{2}\\
    & &\ddots  & &\vdots\\
  & &  &A_{mm}-\lambda I_{n_{m}}  &\xi_{m}  
\end{bmatrix} $ has full row rank. Therefore, we conclude that  the matrix $\left[A-\lambda I_{\sum_{j=1}^{m} n_{j}}\;\;  \xi\right]$ has full row rank for any eigenvalue $\lambda$ of $A$. Hence $(A,\; \xi)$ is controllable.  
\end{proof}

\subsection*{\textbf{Proof of Theorem~\ref{theorem1}}}
\begin{proof}
\emph{Necessity}. Suppose an $N$-mode pure Gaussian state is  generated by a  linear quantum system subject to the two constraints~\ref{constraint1} and~\ref{constraint2}, and $Z$ is the corresponding Gaussian graph matrix for this pure Gaussian state. We will show that the Gaussian graph matrix $Z$ of this pure Gaussian state can be written in the form of Equation~\eqref{thm11} or Equation~\eqref{thm12}.   According to~\ref{constraint1}, the Hamiltonian of the linear quantum system is
$
\hat{H}=\frac{1}{2}\hat{x}^{\top}\begin{bmatrix}
R &0_{N\times N}\\
0_{N\times N}  &R
\end{bmatrix}\hat{x}$, 
where $R=\begin{bmatrix}
2g_{11} &g_{12} &\cdots &g_{1N}\\
g_{12}  &2g_{22} &\cdots &g_{2N}\\
\vdots  & \vdots &\ddots    &\vdots\\
g_{1N}  &g_{2N} &\cdots &2g_{NN}
\end{bmatrix}$. Using Lemma~\ref{lem0}, we have
\begin{numcases}{}
XRX +YRY-\Gamma Y^{-1}X-XY^{-1}\Gamma^{\top}  = R,\label{eq11}\\
-X  R +\Gamma Y^{-1}=0. \label{eq12}
\end{numcases}
It follows from~\eqref{eq12} that $\Gamma=XRY$. Substituting this into~\eqref{eq11} gives 
\begin{align}
YRY-XRX=R. \label{eq13}
\end{align} Recall from Lemma~\ref{lem0} that $\Gamma + \Gamma^{\top}=0$, i.e., $XRY+YRX=0$. Combining this with~\eqref{eq13} gives
\begin{align}
ZRZ=-R. \label{eq14}
\end{align} 
Suppose that the $\ell$th mode of the linear quantum system is locally coupled to the reservoir. Then using Lemma~\ref{lem1}, we have
$ Z_{(\ell,j)}= Z_{(j,\ell)}=0$,  $\forall j\ne \ell$. This fact implies that  there exists a permutation matrix $\mathcal{P}_{1}\in \mathbb{R}^{N \times N}$ such that 
\begin{align*}
Z=\mathcal{P}_{1}^{\top}\begin{bmatrix}
 \bar{z} &0_{1\times (N-1)}\\
0_{(N-1)\times 1} &\breve{Z} 
\end{bmatrix}\mathcal{P}_{1},
\end{align*}
where $\bar{z}=Z_{(\ell,\ell)}$ and $\breve{Z} \in \mathbb{C}^{(N-1)\times (N-1)}$. Obviously, $\bar{z}\in \Lambda$.  Let $\breve{R} \triangleq \mathcal{P}_{1} R \mathcal{P}_{1}^{\top}$. Then Equation~\eqref{eq14} is transformed into
\begin{align}
\begin{bmatrix}
 \bar{z} &0_{1\times (N-1)}\\
0_{(N-1)\times 1} &\breve{Z} 
\end{bmatrix}\breve{R}\begin{bmatrix}
 \bar{z} &0_{1\times (N-1)}\\
0_{(N-1)\times 1} &\breve{Z} 
\end{bmatrix}=-\breve{R}. \label{eq15}
\end{align}
If we write $\breve{R}$ in block form as $\breve{R}=\begin{bmatrix}
\breve{R}_{11} &\breve{R}_{21}^{\top}\\
\breve{R}_{21} &\breve{R}_{22}
\end{bmatrix}$, where $\breve{R}_{11}\in \mathbb{R}$, $\breve{R}_{22}=\breve{R}_{22}^{\top}\in \mathbb{R}^{(N-1)\times (N-1)}$,  and $\breve{R}_{21}\in \mathbb{R}^{(N-1)\times 1}$, then Equation~\eqref{eq15} becomes 
\begin{numcases}{}
\;\; \bar{z} \breve{R}_{11} \bar{z}  = -  \breve{R}_{11}, \notag\\
\breve{Z}\breve{R}_{22}\breve{Z}   =- \breve{R}_{22},\label{eq16}\\
\;\breve{Z}\breve{R}_{21}  \bar{z}  = - \breve{R}_{21}.\label{eq17}
\end{numcases}
Since $\breve{R}_{22}=\breve{R}_{22}^{\top}$, it can be diagonalized by a real orthogonal matrix $\mathcal{Q}_{1}\in \mathbb{R}^{(N-1)\times (N-1)}$. That is, $\breve{R}_{22}= \mathcal{Q}_{1}^{\top} \ohill{R}_{22} \mathcal{Q}_{1}$, where $\ohill{R}_{22}$ is a real diagonal matrix. Let $\ohill{Z} \triangleq  \mathcal{Q}_{1}\breve{Z}\mathcal{Q}_{1}^{\top} $,   and $\ohill{R}_{21}  \triangleq  \mathcal{Q}_{1}\breve{R}_{21} $. Then the equations~\eqref{eq16} and~\eqref{eq17} are transformed into  
\begin{numcases}{}
\ohill{Z}\ohill{R}_{22}\ohill{Z}   =- \ohill{R}_{22},\label{eq18}\\
\;\ohill{Z}\ohill{R}_{21}\bar{z}  = - \ohill{R}_{21}.\label{eq19}
\end{numcases}

Since the $\ell$th mode of the linear quantum system is locally coupled to the reservoir, it follows from Lemma~\ref{lem0} that the matrix $P$ in~\eqref{C} must be of the form 
$P=\begin{bmatrix}
0_{1\times(\ell-1)} &\tau_{p} &0_{1\times(N-\ell)}
\end{bmatrix}^{\top}$, where $\tau_{p} \in \mathbb{C}$ and $\tau_{p}\ne 0$. 
The matrix $Q$ in~\eqref{rankconstraint} is given by
$Q=-iRY+Y^{-1}\Gamma=-iRY+Y^{-1}\left(-YRX\right)=-RZ=-\mathcal{P}_{1}^{\top}\breve{R}\begin{bmatrix}
\bar{z}&0_{1 \times (N-1)}\\
0_{(N-1)\times 1} & \breve{Z} 
\end{bmatrix}\mathcal{P}_{1}$. From~\eqref{rankconstraint}, we observe that the pair $(Q,\; P)$ is controllable. Using Lemma~\ref{lem3},  it follows that $\left(-\breve{R}\begin{bmatrix}
\bar{z} &0_{1 \times  (N-1)}\\
0_{(N-1)\times 1} &\breve{Z}
\end{bmatrix},\; \mathcal{P}_{1}P \right)$ is controllable. Note that $\mathcal{P}_{1}P=\begin{bmatrix}
\tau_{p}  &0_{1\times(N-1)}
\end{bmatrix}^{\top}$ and $-\breve{R}\begin{bmatrix}
\bar{z} &0_{1\times (N-1)}\\
0_{(N-1)\times 1} &\breve{Z}
\end{bmatrix}=-\begin{bmatrix}
\breve{R}_{11}\bar{z} &\breve{R}_{21}^{\top} \breve{Z} \\
\breve{R}_{21}\bar{z}  &\breve{R}_{22} \breve{Z} 
\end{bmatrix}$. It follows from Lemma~\ref{lem2} that 
$\left(-\breve{R}_{22}\breve{Z} ,\; -\breve{R}_{21}\bar{z} \right)$ is controllable. That is, $\left(-\mathcal{Q}_{1}^{\top} \ohill{R}_{22} \ohill{Z} \mathcal{Q}_{1} ,\; -\mathcal{Q}_{1}^{\top}\ohill{R}_{21}\bar{z} \right)$ is controllable. Again using Lemma~\ref{lem3}, it follows that $\left(-\ohill{R}_{22} \ohill{Z}  ,\; -\ohill{R}_{21}\bar{z} \right)$ is controllable. Since $-\ohill{R}_{21}\bar{z}\in \mathbb{C}^{(n-1)\times 1} $, by Lemma~6 in~\cite{MPW15:arxiv}, the matrix $-\ohill{R}_{22} \ohill{Z}$ is a non-derogatory matrix. Then following similar arguments as in the proof of Theorem~1 in~\cite{MPW15:arxiv}, we have the following preliminary result. 
\begin{itemize}
\item[*] If $N$ is odd, then 
\begin{align*}
\ohill{Z}=\mathcal{P}_{2}^{\top}\bar{Z}\mathcal{P}_{2},\quad \bar{Z}=\diag[\tilde{Z}_{1},\cdots,\tilde{Z}_{ \frac{(N-1)}{2}} ],
\end{align*}
 where $\mathcal{P}_{2}\in \mathbb{R}^{(N-1)\times(N-1)}$ is a permutation matrix, $\tilde{Z}_{1}\in\left(\Pi\cup\Xi\right)$, and $\tilde{Z}_{j}\in \Xi$, $2\le j \le  \frac{(N-1)}{2}$. 
 \item[*] If $N$ is even, then 
\begin{align*}
\ohill{Z}=\mathcal{P}_{2}^{\top}\bar{Z}\mathcal{P}_{2},\quad \bar{Z}=\diag[\tilde{Z}_{1},\cdots,\tilde{Z}_{ \frac{N}{2}} ],
\end{align*}
 where $\mathcal{P}_{2}\in \mathbb{R}^{(N-1)\times(N-1)}$ is a permutation matrix, $\tilde{Z}_{1}\in \Lambda $, and $\tilde{Z}_{j}\in \Xi$, $2\le j \le  \frac{N}{2}$. 
\end{itemize}
Here $
\Pi\triangleq\{\diag[z,\;i]\;\;\big|\;\; z\in \mathbb{C}\;\;\text{and}\;\; \im(z)>0 \}$, and $
\Xi\triangleq\bigg\{\mathpzc{Z}\;\;\bigg|\;\; \mathpzc{Z}=\mathpzc{Z}^{\top}\in \mathbb{C}^{2\times 2}, \;\;\im\left(\mathpzc{Z}\right)>0, 
\;\;\text{and}\;\; \big(\diag[1,-1]
\mathpzc{Z}\big)^{2} =-I_{2}\bigg\}$. 

Let $\bar{R}_{22}\triangleq\mathcal{P}_{2} \ohill{R}_{22}\mathcal{P}_{2} ^{\top}  $ and $\bar{R}_{21}\triangleq\mathcal{P}_{2} \ohill{R}_{21}$. Then the equations~\eqref{eq18} and~\eqref{eq19} are transformed into
 \begin{numcases}{}
\bar{Z}\bar{R}_{22}\bar{Z}   =- \bar{R}_{22},\label{eq20}\\
\;\bar{Z}\bar{R}_{21} \bar{z}  = - \bar{R}_{21}.\label{eq21}
\end{numcases}
Since $\left(-\ohill{R}_{22} \ohill{Z}  ,\; -\ohill{R}_{21}\bar{z} \right)$ is controllable, i.e., $\left(-\mathcal{P}_{2}^{\top} \bar{R}_{22} \bar{Z} \mathcal{P}_{2} ,\; -\mathcal{P}_{2}^{\top}\bar{R}_{21}\bar{z} \right)$ is controllable, it follows from Lemma~\ref{lem3} that  $\left(-\bar{R}_{22} \bar{Z}  ,\; -\bar{R}_{21}\bar{z} \right)$ is controllable. 

First, if $N$ is odd,  we partition $\bar{R}_{22}$ and $\bar{R}_{21}$ as $\bar{R}_{22}=\diag[\tilde{R}_{22,1},\cdots,\tilde{R}_{22, \frac{(N-1)}{2}}]$ and $\bar{R}_{21}=\begin{bmatrix}\tilde{R}_{21,1}^{\top} \;\cdots\;\tilde{R}_{21,  \frac{(N-1)}{2}}^{\top}\end{bmatrix}^{\top}$, where $\tilde{R}_{22,j}\in \mathbb{R}^{2\times 2}$ and $\tilde{R}_{21,j}\in \mathbb{R}^{2\times 1}$, $1\le j\le \frac{(N-1)}{2}$. Then the equations~\eqref{eq20} and~\eqref{eq21} become  
 \begin{numcases}{}
\tilde{Z}_{j}\tilde{R}_{22,j}\tilde{Z}_{j}   =- \tilde{R}_{22,j},\label{eq22}\\
\;\tilde{Z}_{j}\tilde{R}_{21,j} \bar{z}  = - \tilde{R}_{21,j},\label{eq23}
\end{numcases}
where $1\le j\le \frac{(N-1)}{2}$.  Since  $\left(-\bar{R}_{22} \bar{Z}  ,\; -\bar{R}_{21}\bar{z} \right)$ is controllable, it follows from Lemma~\ref{lem4}  that $\left(-\tilde{R}_{22,j} \tilde{Z}_{j}  ,\; -\tilde{R}_{21,j}\bar{z} \right)$, $1\le j\le \frac{(N-1)}{2}$, is controllable. Note that $\tilde{R}_{22,1}$ is a real diagonal matrix.   If $\tilde{Z}_{1}\in\Pi$ and $\tilde{Z}_{1}\ne iI_{2}$, solving~\eqref{eq22} gives $\tilde{R}_{22,1}=\begin{bmatrix}
0 &0\\
0 &\tau_{1}
\end{bmatrix}$, where $\tau_{1}\in \mathbb{R}$. Then $-\tilde{R}_{22,1} \tilde{Z}_{1}=\begin{bmatrix}
0 &0\\
0 &-\tau_{1}i
\end{bmatrix}$. Since $\left(-\tilde{R}_{22,1} \tilde{Z}_{1}  ,\; -\tilde{R}_{21,1}\bar{z} \right)$ is controllable, it follows that $\tilde{R}_{21,1}=\begin{bmatrix}
\tau_{2} &\tau_{3}
\end{bmatrix}^{\top}$, with $\tau_{2}\tau_{3}\ne 0$. Substituting this into~\eqref{eq23} yields $\bar{z}=i$ and $\tilde{Z}_{1}= iI_{2}$. This contradicts the assumption that $\tilde{Z}_{1}\ne iI_{2}$. Therefore, we know that if $\tilde{Z}_{1}\in\Pi$, then $\tilde{Z}_{1}= iI_{2}$. Note that $iI_{2} \in \Xi$. Hence $\tilde{Z}_{1}\in\Xi$.  Since $\left(-\tilde{R}_{22,j} \tilde{Z}_{j}  ,\; -\tilde{R}_{21,j}\bar{z} \right)$ is controllable, it follows that $\tilde{R}_{21,j}\ne   0_{2\times  1}$, $1\le j\le \frac{(N-1)}{2}$. Then by~\eqref{eq23}, $-\frac{1}{\bar{z}}$ is an eigenvalue of $\tilde{Z}_{j}$, i.e., $\det\left(\tilde{Z}_{j}+\frac{1}{\bar{z}}\right)=0$, $1\le j \le  \frac{(N-1)}{2}$. Combining this fact with $\tilde{Z}_{j}\in\Xi$, we conclude that $\tilde{Z}_{j}\in\Delta$, $1\le j \le  \frac{(N-1)}{2}$. Let $\mathcal{P}=\mathcal{P}_{1}$, and $\mathcal{Q}=\mathcal{P}_{2}\mathcal{Q}_{1}$  in~\eqref{thm11}. Obviously, $\mathcal{Q}$ is an orthogonal matrix and Equation~\eqref{thm11} holds. This completes the first part of the necessity proof.   
 
Second, if $N$ is even,  we partition $\bar{R}_{22}$ and $\bar{R}_{21}$ as $\bar{R}_{22}=\diag[\tilde{R}_{22,1},\cdots,\tilde{R}_{22, \frac{N}{2}}]$ and $\bar{R}_{21}=\begin{bmatrix}\tilde{R}_{21,1}^{\top}\; \cdots\;\tilde{R}_{21,  \frac{N}{2}}^{\top}\end{bmatrix}^{\top}$, where $\tilde{R}_{22,1}\in \mathbb{R}$, $\tilde{R}_{21,1}\in \mathbb{R}$,  $\tilde{R}_{22,j}\in \mathbb{R}^{2\times 2}$ and $\tilde{R}_{21,j}\in \mathbb{R}^{2\times 1}$, $2\le j\le \frac{N}{2}$. Then the equations~\eqref{eq20} and~\eqref{eq21} become
 \begin{numcases}{}
\tilde{Z}_{j}\tilde{R}_{22,j}\tilde{Z}_{j}   =- \tilde{R}_{22,j},\label{eq24}\\
\;\tilde{Z}_{j}\tilde{R}_{21,j} \bar{z}  = - \tilde{R}_{21,j},\label{eq25}
\end{numcases}
where $1\le j\le \frac{N}{2}$.  Recall that  $\left(-\bar{R}_{22} \bar{Z}  ,\; -\bar{R}_{21}\bar{z} \right)$ is controllable. Then it follows from Lemma~\ref{lem4} that  $\left(-\tilde{R}_{22,j} \tilde{Z}_{j},\; -\tilde{R}_{21,j}\bar{z} \right)$, $1\le j\le \frac{N}{2}$, is controllable. Hence we have $\tilde{R}_{21,1}\ne 0$ and $\tilde{R}_{21,j}\ne 0_{2\times 1}$, $2\le j\le \frac{N}{2}$. Then by~\eqref{eq25}, $-\frac{1}{\bar{z}}$ is an eigenvalue of $\tilde{Z}_{j}$, i.e., $\det\left(\tilde{Z}_{j}+\frac{1}{\bar{z}}\right)=0$, $1\le j \le  \frac{N}{2}$. Therefore, we obtain $ \tilde{Z}_{1}=-\frac{1}{\bar{z}}$ and $\tilde{Z}_{j}\in\Delta$, $2\le j \le  \frac{N}{2}$. Let $\mathcal{P}=\mathcal{P}_{1}$, and $\mathcal{Q}=\mathcal{P}_{2}\mathcal{Q}_{1}$  in~\eqref{thm12}. Obviously, $\mathcal{Q}$ is an orthogonal matrix and Equation~\eqref{thm12} holds. This completes the second part of the necessity proof.   
 
 \emph{Sufficiency}. Suppose the graph matrix $Z$ of an $N$-mode pure Gaussian state (where $N$ is odd) satisfies Equation~\eqref{thm11}. We now construct an $N$-mode linear quantum system subject to the two constraints~\ref{constraint1} and~\ref{constraint2}, such that the state is obtained as the unique steady state of this system. We see from Equation~\eqref{thm11} that $\tilde{Z}_{j}\in \Delta$, $1\le j\le \frac{(N-1)}{2}$. Using Theorem~2 in~\cite{MPW15:arxiv}, it can be shown that $\xi_{v1}=\begin{bmatrix}
 1 &1
 \end{bmatrix}^{\top}$ and $\xi_{v2}=\begin{bmatrix}
 1 &-1
 \end{bmatrix}^{\top}$ are two eigenvectors of $\tilde{Z}_{j}$, $1\le j\le \frac{(N-1)}{2}$. Since $\tilde{Z}_{j}\in\mathbb{C}^{2\times 2}$ and $-\frac{1}{\bar{z}}$ is an eigenvalue of $\tilde{Z}_{j}$,  we have $\tilde{Z}_{j}\xi_{v1}=-\frac{1}{\bar{z}}\xi_{v1}$ or  $\tilde{Z}_{j}\xi_{v2}=-\frac{1}{\bar{z}}\xi_{v2}$. Using this fact, we choose $\tilde{R}_{21,j}=\bar{\tau}_{j}\begin{bmatrix}
 1 &1
 \end{bmatrix}^{\top}$,  where $\bar{\tau}_{j}\in \mathbb{R}$ and $\bar{\tau}_{j}\ne 0$,  if $\tilde{Z}_{j}\begin{bmatrix}
 1 &1
 \end{bmatrix}^{\top}=-\frac{1}{\bar{z}}\begin{bmatrix}
 1 &1
 \end{bmatrix}^{\top}$ and $\tilde{R}_{21,j}=\bar{\tau}_{j}\begin{bmatrix}
 1 &-1
 \end{bmatrix}^{\top}$ otherwise. Then we have $\tilde{Z}_{j}\tilde{R}_{21,j}=-\frac{1}{\bar{z}}\tilde{R}_{21,j}$, $1\le j\le \frac{(N-1)}{2}$.  Let $
 \bar{R}_{21}=\begin{bmatrix}\tilde{R}_{21,1}^{\top} \;\cdots\;\tilde{R}_{21,  \frac{(N-1)}{2}}^{\top}\end{bmatrix}^{\top}$
 and let $
 \bar{R}_{22}=\diag\left[
r_{1},\;\; 
-r_{1},\;\; 
r_{2},\; \;
-r_{2},\;\; 
\cdots ,\;\; 
r_{\frac{N-1}{2}},\; \;
-r_{\frac{N-1}{2}}\right]$, where $r_{j}\in \mathbb{R}$, $r_{j} \ne 0$, and $|r_{j}|\ne |r_{k}|$  whenever $ j\ne k$. Then it can be verified that 
$\bar{Z}\bar{R}_{22}\bar{Z}   =- \bar{R}_{22}$, and $\bar{Z}\bar{R}_{21} \bar{z}  = - \bar{R}_{21}$. In Lemma~\ref{lem0}, let us choose $R=\mathcal{P}^{\top}\begin{bmatrix}
0  &\bar{R}_{21}^{\top}\mathcal{Q}\\
\mathcal{Q}^{\top}\bar{R}_{21} & \mathcal{Q}^{\top}\bar{R}_{22}\mathcal{Q}
\end{bmatrix}\mathcal{P}$, $\Gamma=XRY$ and $P=\mathcal{P}^{\top}\begin{bmatrix}\tau_{p} &0_{1\times (N-1)}\end{bmatrix}^{\top}$ where $\tau_{p}\in \mathbb{C}$ and $\tau_{p}\ne 0$.  Then it can be verified that $R=R^{\top}\in\mathbb{R}^{N\times N}$ and $ZRZ=-R$. It then follows that $YRY-XRX=R$ and $XRY+YRX=0$. Hence $\Gamma+\Gamma^{\top}=0$, i.e., $\Gamma$ is a skew symmetric matrix. 
 Substituting $R$ and $\Gamma$ into~\eqref{G} gives $G=\diag[R,\;\;R]$. Therefore, the Hamiltonian is  $\hat{H}=\frac{1}{2}\hat{x}^{\top}G\hat{x}$, which  satisfies the first constraint~\ref{constraint1}. We have 
 \begin{align*}
Q&=-RZ\\
&=-\mathcal{P}^{\top}\begin{bmatrix}
0  &\bar{R}_{21}^{\top}\mathcal{Q}\\
\mathcal{Q}^{\top}\bar{R}_{21} & \mathcal{Q}^{\top}\bar{R}_{22}\mathcal{Q}
\end{bmatrix}\begin{bmatrix}
\bar{z}  &0_{ 1 \times (N-1)}\\
0_{(N-1) \times 1} & \mathcal{Q}^{\top}\bar{Z}\mathcal{Q}
\end{bmatrix}\mathcal{P}  \\
&=-\mathcal{P}^{\top}\begin{bmatrix}
0  &\bar{R}_{21}^{\top} \bar{Z}\mathcal{Q}\\
\mathcal{Q}^{\top}\bar{R}_{21} \bar{z}  & \mathcal{Q}^{\top}\bar{R}_{22}\bar{Z}\mathcal{Q}
\end{bmatrix}\mathcal{P}. 
\end{align*}
We now show that $(Q,P)$ is controllable. Let $\tilde{R}_{22,j}=\diag[r_{j},\;\;-r_{j}]$. We have
\begin{align*}
&\rank\left(\left[\tilde{R}_{21,j} \bar{z} \;\;\;\;\tilde{R}_{22,j}\tilde{Z}_{j}\tilde{R}_{21,j} \bar{z} \right]\right)\\
=&\rank\left(\left[\tilde{R}_{21,j} \bar{z} \;\;\;\;-\tilde{R}_{22,j}\tilde{R}_{21,j}\right]\right)\\
=&\rank\left(\left[\tilde{R}_{21,j} \;\;\;\;\tilde{R}_{22,j}\tilde{R}_{21,j}\right]\right)\\
=& \rank\left(\begin{bmatrix}
\bar{\tau}_{j} &r_{j}\bar{\tau}_{j} \\
\bar{\tau}_{j} &-r_{j}\bar{\tau}_{j}
\end{bmatrix}\right) \; \text{or} \;\rank\left(\begin{bmatrix}
\bar{\tau}_{j} &r_{j}\bar{\tau}_{j} \\
-\bar{\tau}_{j} &r_{j}\bar{\tau}_{j}
\end{bmatrix}\right)\\
=&2,\quad 1\le j\le \frac{(N-1)}{2}. 
\end{align*}
Then it follows that  $\left(\tilde{R}_{22,j}\tilde{Z}_{j}, \tilde{R}_{21,j} \bar{z}  \right)$, $1\le j\le \frac{(N-1)}{2}$, is controllable. In addition, we have $(\tilde{R}_{22,j}\tilde{Z}_{j})^{2}=-\tilde{R}_{22,j}^{2}=-r_{j}^{2}I_{2}$. Using Lemma~2 in~\cite{MPW15:arxiv}, it follows that the matrix $ \tilde{R}_{22,j}\tilde{Z}_{j} $ is diagonalizable and its eigenvalues are either $r_{j}i$ or $-r_{j}i$. Hence  $\tilde{R}_{22,j}\tilde{Z}_{j}$, $1\le j\le \frac{(N-1)}{2}$, have no common eigenvalues.  By Lemma~\ref{lem4}, we have $\left(\bar{R}_{22}\bar{Z}, \bar{R}_{21} \bar{z}  \right)$ is controllable. By Lemma~\ref{lem3}, $\left(\mathcal{Q}^{\top}\bar{R}_{22}\bar{Z}\mathcal{Q}, \mathcal{Q}^{\top}\bar{R}_{21} \bar{z}  \right)$ is controllable. According to Lemma~\ref{lem2}, $\left(\begin{bmatrix}
0  &\bar{R}_{21}^{\top} \bar{Z}\mathcal{Q}\\
\mathcal{Q}^{\top}\bar{R}_{21} \bar{z}  & \mathcal{Q}^{\top}\bar{R}_{22}\bar{Z}\mathcal{Q}
\end{bmatrix},\begin{bmatrix}\tau_{p} \\ 0_{(N-1)\times 1 }\end{bmatrix} \right)$ is controllable. Based on Lemma~\ref{lem3}, we conclude that $(Q,P)$ is controllable. Therefore, by Lemma~\ref{lem0}, the resulting linear quantum system is strictly stable and  generates the desired target pure Gaussian state specified by Equation~\eqref{thm11}. The system–-reservoir coupling vector $\hat{L}$    is given by 
\begin{align*}
\hat{L}= C\hat{x}&=P^{\top}\left[-Z\;\; I_{N}\right]\hat{x} \\
& = - \begin{bmatrix}\tau_{p}\bar{z} &0_{1\times (N-1)}\end{bmatrix}  \mathcal{P} \begin{bmatrix}
\hat{q}_{1} &\hat{q}_{2} &\cdots &\hat{q}_{N}
\end{bmatrix}^{\top}\\
&\quad + \begin{bmatrix}\tau_{p} &0_{1\times (N-1)}\end{bmatrix} \mathcal{P} \begin{bmatrix}
\hat{p}_{1} &\hat{p}_{2} &\cdots &\hat{p}_{N}
\end{bmatrix}^{\top}\\
& = -\tau_{p}\bar{z}\hat{q}_{\ell}+\tau_{p}\hat{p}_{\ell}, 
\end{align*}
where $\ell = \begin{bmatrix}1 &0_{1\times (N-1)}\end{bmatrix}\mathcal{P} \begin{bmatrix}
1 &2 &\cdots &N
\end{bmatrix}^{\top} $.  
It can be seen that the system–-reservoir coupling vector $\hat{L}$  consists of only one Lindblad operator which acts on the $\ell$th mode of the system. Hence it satisfies the second constraint~\ref{constraint2}. This completes the first part of the sufficiency proof.

Suppose the graph matrix $Z$ of an $N$-mode pure Gaussian state (where $N$ is even) satisfies Equation~\eqref{thm12}. We now construct an $N$-mode linear quantum system subject to the two constraints~\ref{constraint1} and~\ref{constraint2}, such that the state is obtained as the unique steady state of this system. We see from Equation~\eqref{thm12} that $\tilde{Z}_{1}=-\frac{1}{\bar{z}}$ and $\tilde{Z}_{j}\in \Delta$, $2\le j\le \frac{N }{2}$. Using Theorem~2 in~\cite{MPW15:arxiv}, it can be shown that $\xi_{v1}=\begin{bmatrix}
1 &1
 \end{bmatrix}^{\top}$ and $\xi_{v2}=\begin{bmatrix}
1 &-1
 \end{bmatrix}^{\top}$ are two eigenvectors of $\tilde{Z}_{j}$, $2\le j\le \frac{N}{2}$. Since $\tilde{Z}_{j}\in\mathbb{C}^{2\times 2}$ and $-\frac{1}{\bar{z}}$ is an eigenvalue of $\tilde{Z}_{j}$,  we have $\tilde{Z}_{j}\xi_{v1}=-\frac{1}{\bar{z}}\xi_{v1}$ or  $\tilde{Z}_{j}\xi_{v2}=-\frac{1}{\bar{z}}\xi_{v2}$, $2\le j\le \frac{N}{2}$. Let $\tilde{R}_{21,1}=\bar{\tau}_{1}$, where $\bar{\tau}_{1} \in \mathbb{R}$ and $\bar{\tau}_{1}\ne 0$. For $2\le j\le \frac{N}{2}$, let $\tilde{R}_{21,j}=\bar{\tau}_{j}\begin{bmatrix}
1 &1
 \end{bmatrix}^{\top}$, where $\bar{\tau}_{j} \in \mathbb{R}$ and $\bar{\tau}_{j}\ne 0$, if $\tilde{Z}_{j}\begin{bmatrix}
1 &1
 \end{bmatrix}^{\top}=-\frac{1}{\bar{z}}\begin{bmatrix}
1 &1
 \end{bmatrix}^{\top}$ and $\tilde{R}_{21,j}=\bar{\tau}_{j}\begin{bmatrix}
1 &-1
 \end{bmatrix}^{\top}$ otherwise. Then we have $\tilde{Z}_{j}\tilde{R}_{21,j}=-\frac{1}{\bar{z}}\tilde{R}_{21,j}$, $1\le j\le \frac{ N }{2}$.  Let $
 \bar{R}_{21}=\begin{bmatrix}\tilde{R}_{21,1}^{\top} \;\cdots\;\tilde{R}_{21,  \frac{N}{2}}^{\top}\end{bmatrix}^{\top}$ and let  $
 \bar{R}_{22}=\diag\left[0,\;\;
r_{2},\; \;
-r_{2},\;\; 
\cdots ,\;\; 
r_{\frac{N}{2}},\; \;
-r_{\frac{N}{2}}\right]$, where $r_{j}\in \mathbb{R}$, $r_{j} \ne 0$, $j\ge 2$, and $|r_{j}|\ne |r_{k}|$  whenever $ j\ne k$.  Then it can be verified that 
$\bar{Z}\bar{R}_{22}\bar{Z}   =- \bar{R}_{22}$  and $\bar{Z}\bar{R}_{21} \bar{z}  = - \bar{R}_{21}$. Let $R=\mathcal{P}^{\top}\begin{bmatrix}
0  &\bar{R}_{21}^{\top}\mathcal{Q}\\
\mathcal{Q}^{\top}\bar{R}_{21} & \mathcal{Q}^{\top}\bar{R}_{22}\mathcal{Q}
\end{bmatrix}\mathcal{P}$, $\Gamma=XRY$ and $P=\mathcal{P}^{\top}\begin{bmatrix}\tau_{p} &0_{1\times (N-1)}\end{bmatrix}^{\top}$ where $\tau_{p}\in \mathbb{C}$ and $\tau_{p}\ne 0$.  Then it can be verified that $R=R^{\top}\in\mathbb{R}^{N\times N}$  and $ZRZ=-R$. It then follows that $YRY-XRX=R$ and $XRY+YRX=0$. Hence $\Gamma+\Gamma^{\top}=0$, i.e., $\Gamma$ is a skew symmetric matrix. 
 Substituting $R$ and $\Gamma$ into~\eqref{G} gives $G=\diag[R,\;\;R]$. Therefore, the Hamiltonian is  $\hat{H}=\frac{1}{2}\hat{x}^{\top}G\hat{x}$, which  satisfies the first constraint~\ref{constraint1}.  We have 
 \begin{align*}
Q&=-RZ\\
&=-\mathcal{P}^{\top}\begin{bmatrix}
0  &\bar{R}_{21}^{\top}\mathcal{Q}\\
\mathcal{Q}^{\top}\bar{R}_{21} & \mathcal{Q}^{\top}\bar{R}_{22}\mathcal{Q}
\end{bmatrix}\begin{bmatrix}
\bar{z}  &0_{ 1 \times (N-1)}\\
0_{(N-1) \times 1} & \mathcal{Q}^{\top}\bar{Z}\mathcal{Q}
\end{bmatrix}\mathcal{P}  \\
&=-\mathcal{P}^{\top}\begin{bmatrix}
0  &\bar{R}_{21}^{\top} \bar{Z}\mathcal{Q}\\
\mathcal{Q}^{\top}\bar{R}_{21} \bar{z}  & \mathcal{Q}^{\top}\bar{R}_{22}\bar{Z}\mathcal{Q}
\end{bmatrix}\mathcal{P}. 
\end{align*}
We now show that $(Q,P)$ is controllable.   Let $\tilde{R}_{22,1}=0$   and $\tilde{R}_{22,j}=\diag[r_{j},\;\;-r_{j}]$, $2\le j\le \frac{N}{2}$. First, it can be seen that $\left(\tilde{R}_{22,1}\tilde{Z}_{1}, \tilde{R}_{21,1} \bar{z}  \right)$  is  controllable. We also have 
\begin{align*}
&\rank\left(\left[\tilde{R}_{21,j} \bar{z} \;\;\;\;\tilde{R}_{22,j}\tilde{Z}_{j}\tilde{R}_{21,j} \bar{z} \right]\right)\\
=&\rank\left(\left[\tilde{R}_{21,j} \bar{z} \;\;\;\;-\tilde{R}_{22,j}\tilde{R}_{21,j}\right]\right)\\
=&\rank\left(\left[\tilde{R}_{21,j} \;\;\;\;\tilde{R}_{22,j}\tilde{R}_{21,j}\right]\right)\\
=&\rank\left(\begin{bmatrix}
\bar{\tau}_{j} &r_{j}\bar{\tau}_{j} \\
\bar{\tau}_{j} &-r_{j}\bar{\tau}_{j}
\end{bmatrix}\right) \; \text{or} \;\rank\left(\begin{bmatrix}
\bar{\tau}_{j} &r_{j}\bar{\tau}_{j} \\
-\bar{\tau}_{j} &r_{j}\bar{\tau}_{j}
\end{bmatrix}\right)\\
=&2,\quad 2\le j\le \frac{N}{2}. 
\end{align*}
Then it follows that  $\left(\tilde{R}_{22,j}\tilde{Z}_{j}, \tilde{R}_{21,j} \bar{z}  \right)$, $2\le j\le \frac{N}{2}$, is controllable. In addition, we have $(\tilde{R}_{22,j}\tilde{Z}_{j})^{2}=-\tilde{R}_{22,j}^{2}=-r_{j}^{2}I_{2}$, $2\le j\le \frac{N}{2}$. Using Lemma~2 in~\cite{MPW15:arxiv}, it follows that the matrix $ \tilde{R}_{22,j}\tilde{Z}_{j} $ is diagonalizable and its eigenvalues are either $r_{j}i$ or $-r_{j} i$, $2\le j\le \frac{N}{2}$. Bearing in mind $\tilde{R}_{22,1}\tilde{Z}_{1}=0$, it follows that  $\tilde{R}_{22,j}\tilde{Z}_{j}$, $1\le j\le \frac{N}{2}$, have no common eigenvalues. By  Lemma~\ref{lem4}, $\left(\bar{R}_{22}\bar{Z}, \bar{R}_{21} \bar{z}  \right)$ is controllable. It then follows from Lemma~\ref{lem3} that $\left(\mathcal{Q}^{\top}\bar{R}_{22}\bar{Z}\mathcal{Q}, \mathcal{Q}^{\top}\bar{R}_{21} \bar{z}  \right)$ is controllable. It follows from Lemma~\ref{lem2} that $\left(\begin{bmatrix}
0  &\bar{R}_{21}^{\top} \bar{Z}\mathcal{Q}\\
\mathcal{Q}^{\top}\bar{R}_{21} \bar{z}  & \mathcal{Q}^{\top}\bar{R}_{22}\bar{Z}\mathcal{Q}
\end{bmatrix},\begin{bmatrix}\tau_{p} &0_{1\times (N-1)}\end{bmatrix}^{\top}\right)$ is controllable. According to Lemma~\ref{lem3}, we conclude that $(Q,P)$ is controllable. Therefore, by Lemma~\ref{lem0}, the resulting linear quantum system is strictly stable and generates the desired target pure Gaussian state specified by Equation~\eqref{thm12}. The system–-reservoir coupling vector $\hat{L}$   is given by 
\begin{align*}
\hat{L}= C\hat{x}&=P^{\top}\left[-Z\;\; I_{N}\right]\hat{x} \\
& = - \begin{bmatrix}\tau_{p}\bar{z} &0_{1\times (N-1)}\end{bmatrix}  \mathcal{P} \begin{bmatrix}
\hat{q}_{1} &\hat{q}_{2} &\cdots &\hat{q}_{N}
\end{bmatrix}^{\top} \\
& \quad + \begin{bmatrix}\tau_{p} &0_{1\times (N-1)}\end{bmatrix} \mathcal{P} \begin{bmatrix}
\hat{p}_{1} &\hat{p}_{2} &\cdots &\hat{p}_{N}
\end{bmatrix}^{\top}\\
& = -\tau_{p}\bar{z}\hat{q}_{\ell}+\tau_{p}\hat{p}_{\ell}, 
\end{align*}
where $\ell = \begin{bmatrix}1 &0_{1\times (N-1)}\end{bmatrix}\mathcal{P} \begin{bmatrix}
1  &2 &\cdots &N
\end{bmatrix}^{\top} $.  
It can be seen that the system–-reservoir coupling vector $\hat{L}$  consists of only one Lindblad operator which acts on the $\ell$th mode of the system. Hence it satisfies the second constraint~\ref{constraint2}. This completes the second part of the sufficiency proof. 
\end{proof}

\subsection*{\textbf{Proof of Lemma~\ref{theorem2}}}
\begin{proof}
Suppose $\mathpzc{Z}\in \Delta $. Then it follows from Theorem~2 in~\cite{MPW15:arxiv} that  $\mathpzc{Z}$ has the form $\mathpzc{Z}=\begin{bmatrix}z_{11} &z_{12}\\z_{12} &z_{11} \end{bmatrix}$, where $z_{11}\in \mathbb{C}$ and $z_{12}\in \mathbb{C}$. Substituting this into the equation $\bigg(\diag[1,-1]
 \mathpzc{Z} \bigg)^{2} =-I_{2}$ gives
 \begin{align}
 z_{11}^{2}-z_{12}^{2}=-1.  \label{thm2pf1}
 \end{align}
 The constraint $\det(\mathpzc{Z}+ \frac{1}{\bar{z}}I_{2})=0$  is equivalent to
 \begin{align}
 \left(z_{11}+\frac{1}{\bar{z}}\right)^{2}-z_{12}^{2}=0. \label{thm2pf2}
 \end{align}
 Combining~\eqref{thm2pf1} and~\eqref{thm2pf2} yields $z_{11}=\frac{\bar{z}^{2}-1}{2\bar{z}}$ and $z_{12}=\pm \frac{\bar{z}^{2}+1}{2\bar{z}}$. Therefore, $\mathpzc{Z}=\begin{bmatrix}\frac{\bar{z}^{2}-1}{2\bar{z}} &\frac{\bar{z}^{2}+1}{2\bar{z}}\\ \frac{\bar{z}^{2}+1}{2\bar{z}} &\frac{\bar{z}^{2}-1}{2\bar{z}} \end{bmatrix}$ or $\begin{bmatrix}\frac{\bar{z}^{2}-1}{2\bar{z}} &-\frac{\bar{z}^{2}+1}{2\bar{z}}\\ -\frac{\bar{z}^{2}+1}{2\bar{z}} &\frac{\bar{z}^{2}-1}{2\bar{z}} \end{bmatrix}$. Since $\mathpzc{Z}= \mathpzc{Z} ^{\top}$, it  can be easily seen that  $\re(\mathpzc{Z})=\re(\mathpzc{Z})^{\top}$ and $\im(\mathpzc{Z})=\im(\mathpzc{Z})^{\top}$. To prove $\im(\mathpzc{Z})>0$, we have to show $\im(z_{11})>0$   and $\left(\im(z_{11})\right)^{2}-\left(\im(z_{12})\right)^{2}>0$. Suppose $\bar{z}=x+iy$, $x \in \mathbb{R}$, $y \in \mathbb{R}$, and $y >0$. Then 
 \begin{align*}
 z_{11}=\frac{\bar{z}^{2}-1}{2\bar{z}}=\frac{1}{2}\left(\bar{z}-\frac{ 1}{\bar{z}}\right)=\frac{1}{2}\left(x+iy-\frac{x-iy}{x^{2}+y^{2}}\right). 
 \end{align*}
 Hence $\im(z_{11})=\frac{1}{2}\left(y+\frac{y}{x^{2}+y^{2}}\right) >0$. Since  
  \begin{align*}
 z_{12}=\pm\frac{\bar{z}^{2}+1}{2\bar{z}}=\pm\frac{1}{2}\left(\bar{z}+\frac{ 1}{\bar{z}}\right)=\pm\frac{1}{2}\left(x+iy+\frac{x-iy}{x^{2}+y^{2}}\right), 
 \end{align*}
 we have 
\begin{align*} 
&\left(\im(z_{11})\right)^{2}-\left(\im(z_{12})\right)^{2}\\
= &
 \frac{1}{4}\left(y+\frac{y}{x^{2}+y^{2}}\right)^{2}-\frac{1}{4}\left(y-\frac{y}{x^{2}+y^{2}}\right)^{2}\\
 = & \frac{y^{2}}{\left(x^{2}+y^{2}\right)^{2}}>0. 
 \end{align*}
Therefore, we have $\im(\mathpzc{Z})>0$. Combining the results above, we conclude that $\Delta =\left\{\begin{bmatrix}\frac{\bar{z}^{2}-1}{2\bar{z}} &\frac{\bar{z}^{2}+1}{2\bar{z}}\\ \frac{\bar{z}^{2}+1}{2\bar{z}} &\frac{\bar{z}^{2}-1}{2\bar{z}} \end{bmatrix},\quad\begin{bmatrix}\frac{\bar{z}^{2}-1}{2\bar{z}} &-\frac{\bar{z}^{2}+1}{2\bar{z}}\\ -\frac{\bar{z}^{2}+1}{2\bar{z}} &\frac{\bar{z}^{2}-1}{2\bar{z}} \end{bmatrix} \right\}$. This completes the proof. 
\end{proof}

\end{document}